\documentclass[aps,preprint,groupedaddress]{revtex4-1}
\usepackage{amssymb}
\usepackage{amsmath}
\usepackage{amsthm}
\usepackage{graphicx}
\usepackage{multirow}
\usepackage{booktabs}
\usepackage{units}
\graphicspath{{./figures/}}
\newtheorem{lemma}{Lemma}
\newtheorem{theorem}{Theorem}
\newtheorem{problem}{Problem}
\begin{document}
\title{Randomisation Algorithms for Large Sparse Matrices}

\author{Kai Puolam\"aki}
\email{kai.puolamaki@helsinki.fi}

\affiliation{Aalto University, Helsinki, Finland}
\affiliation{Department of Computer Science, University of Helsinki, Finland}

\author{Andreas Henelius}
\email{andreas.henelius@helsinki.fi}

\affiliation{Finnish Institute of Occupational Health, Helsinki, Finland}
\affiliation{Aalto University, Helsinki, Finland}
\affiliation{Department of Computer Science, University of Helsinki, Finland}

\author{Antti Ukkonen}
\email{antti.ukkonen@helsinki.fi}

\affiliation{Department of Computer Science, University of Helsinki, Finland}

\date{14 November 2018}
\begin{abstract}
In many domains it is necessary to generate surrogate networks, e.g.,
for hypothesis testing of different properties of a
network. Furthermore, generating surrogate networks typically requires
that different properties of the network is preserved, e.g., edges may
not be added or deleted and the edge weights may be restricted to
certain intervals. In this paper we introduce a novel efficient
property-preserving Markov Chain Monte Carlo method termed
CycleSampler for generating surrogate networks in which (i) edge
weights are constrained to an interval and node weights are preserved
exactly, and (ii) edge and node weights are both constrained to
intervals. These two types of constraints cover a wide variety of
practical use-cases. The method is applicable to both undirected and
directed graphs. We empirically demonstrate the efficiency of the
CycleSampler method on real-world datasets. We provide an
implementation of CycleSampler in R, with parts implemented in C.
\end{abstract}
\maketitle

\section{Introduction}
In many applications it is useful to represent relationships between
objects with a network in which vertices correspond to objects of
interest and associations between objects are expressed with directed
or undirected edges. The edges can also be weighted. Given such a
network, one might be interested in questions such as community
detection \cite{fortunato:2010:a}, clustering coefficients
\cite{ansmann2011constrained, hao:2012:a}, centrality measures
\cite{joyce:2010:a}, shortest path distributions
\cite{BackstromBRUV12}, or different measures of information
propagation \cite{chen:2013:a}. However, it is often useful to study
whether a possibly interesting finding from a given network reflects a
real phenomenon, or if it is merely caused by noise. A simple approach
to this is to compare the original finding to findings from {\em
  surrogate networks} that share some {\em relevant properties} with
the original network, but are otherwise inherently ``random''.  For
example, communities found in the original network should probably
exhibit greater structure than communities in appropriately randomised
networks. Usual solutions thus involve generating a number of
surrogate networks by fixing some network properties of interest, and
then drawing a uniform sample from the set of all networks satisfying
the given properties.

Existing methods for sampling networks can be assigned into two
categories: {\em property-preserving} and {\em structure-preserving}
methods.  Property-preserving methods \cite{coolen2009constrained,
  Hanhijarvi2009graph, ansmann2011constrained, roberts2012unbiased} do
not preserve network topology, i.e., they can introduce new edges and
remove existing ones.  These approaches can be viewed as always
considering a fully connected clique, inside which the edge weights
are rearranged.  Their aim is to preserve some network property
of interest, such as node degrees or node weights.
Property-preserving methods can be further divided into those
preserving the property {\em exactly} or {\em in expectation}.  For
example, preserving node degrees exactly is relatively straightforward
using, e.g., edge swaps \cite{Gionis2007,ansmann2011constrained}.
Preserving higher-order statistics is often possible only in
expectation \cite{Hanhijarvi2009graph}. That is, the expected value of
the property remains equal to some given constraint, but its observed
value in an individual sample may deviate from this constraint.

Structure-preserving methods \cite{squartini2011analytical}, on the
other hand, keep the network topology fixed (new edges are not
inserted and existing ones are not removed), but usually maintain the
desired property, e.g., node weights, only in expectation. Such
approaches are usually based on {\em maximum entropy} models where
surrogate networks are sampled simply by drawing edge weights from a
parametrized i.i.d.~distribution. While these methods are often
computationally quite efficient, maintaining the desired property only
in expectation may not be enough. It is, e.g., conceivable that
without additional constraints the network property of interest may
occasionally take values that cannot be observed in real networks, and
in these cases the sampled node weights may hence be unsuitable for
the task of comparing an original finding to ``random findings''.

As a toy example, consider the network shown in
Fig.~\ref{fig:example1:simple}, with six nodes and seven edges. This
network describes telephone calls between six individuals (the nodes)
over an observation period of 24 hours. An edge between two nodes
represents the total cumulative call duration (in hours) between two
individuals. Our goal is now to sample networks having exactly the
same edges as in the original network, i.e., we assume that people
cannot communicate outside their friendship network and hence no new
edges are created. We consider two different cases. First, we consider
the case where the call durations between two individuals (the edge
weights) can vary within an interval, but the node weights (sum of
edge weights adjacent to a node) must stay fixed at their original
values.  That is, the total duration spent on the phone by every
individual must remain the same. Second, we consider the case where
both edge and node weights are allowed to vary within a given
interval.  In both cases the interval width is constrained by the
simple fact that during a 24-hour period a person cannot spend more
than 24 hours on the phone in total.

Table~\ref{tab:example1:weights} shows the range (min and max) of the
edge and node weights for 10 000 samples for the two cases described
above, obtained using (i) the method presented in this paper, termed
\emph{CycleSampler}, and (ii) a maximum entropy model. In Case 1, the
CycleSampler method preserves the node weights exactly.  (Notice that
the range of node weights matches the node weights given in
Fig.~\ref{fig:example1:simple}.)  In Case 2 the CycleSampler method
simultaneously preserves both edge and node weights within the
interval 0--24 hours, while the maximum entropy method preserves these
only in expectation. It is clear that the maximum entropy model easily
satisfies constraints on edge weights, but violates the 24-hour node
weight constraint.  This is because the edge weights are sampled
i.i.d.~and thus for nodes having a large degree the sum of sampled
weights on adjacent edges can easily increase beyond the maximum value
allowed.

\begin{figure}
  \centering
  \begin{minipage}[t]{0.48\textwidth}
    \includegraphics[width=\textwidth]{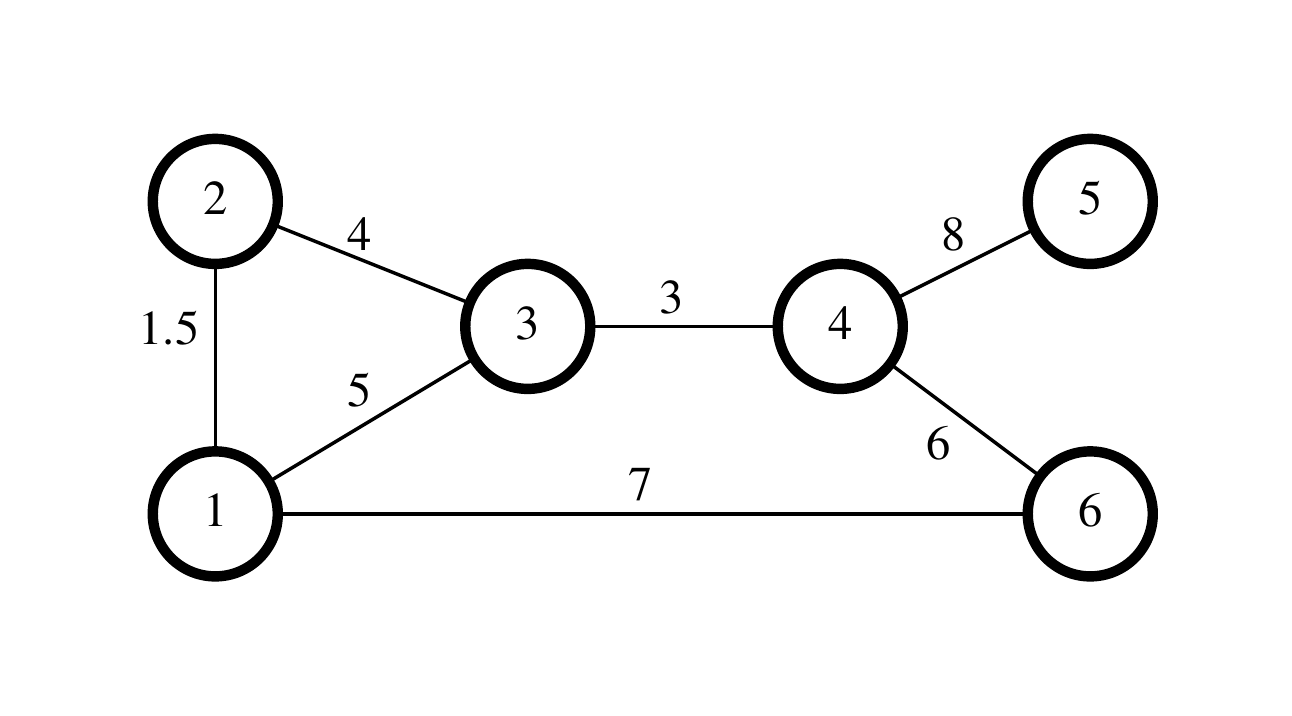}
  \end{minipage}
  \begin{minipage}[t]{0.48\textwidth}
    \vspace*{-10em}
    \setlength{\tabcolsep}{1.7ex}
    \begin{tabular}{rcccccc}
      \multicolumn{7}{c}{Node weights} \\
      \textbf{Node}   & 1    & 2   & 3    & 4    & 5   & 6 \\
      \textbf{Weight} & 13.5 & 5.5 & 12.0 & 17.0 & 8.0 & 13.0
    \end{tabular}
  \end{minipage}
  \caption{Network where persons are represented as nodes and phone
    calls between persons as edges, with the edge weight denoting
    the total call duration in hours between two persons. The table
    on the right shows the node weights (sum of adjacent edge weights)
    for each node of the network.}
  \label{fig:example1:simple}
\end{figure}

\begin{table}
  \centering
  \caption{Range of edge and node weights of 10 000 surrogate networks
    sampled using the Uniform sampling method presented in this paper
    and the maximum entropy method. In Case 1 node weights are fixed,
    while in Case 2 edge and node weights can vary. Note that edge and
    node weights stay within the physically feasible interval [0, 24]
    hours for the uniform sampling method, whereas the maximum entropy
    solution can lead to unfeasible node weights.}
  \label{tab:example1:weights}
  \begin{ruledtabular}
    \begin{tabular}{r cccc}
      & \multicolumn{2}{c}{\textbf{Case 1}} & \multicolumn{2}{c}{\textbf{Case 2}} \\
      \cline{2-3} \cline{4-5}
      \textbf{Method} & \textbf{edge weights} & \textbf{node weights} & \textbf{edge weights} & \textbf{node weights}\\
      \hline
      CycleSampler    & [0.00, 12.00] & [5.50, 17.00] & [0.00, 22.28] & [0.06, 24.00] \\ 
      Maximum Entropy & & & [0.00, 24.00] & [0.00, 54.90  \\ 
    \end{tabular}
  \end{ruledtabular}
\end{table}

This limitation of the maximum entropy model is further highlighted in
the numerical example presented in Figure~\ref{fig:toy3}. This example
shows that even in a very simple case the majority of samples from a
maximum entropy model will {\em not} satisfy hard constraints on node
weights, even if the expected value of node weights is preserved.
Furthermore, the resulting distribution of edge weights is not
uniform.  The CycleSampler method proposed in this paper, on the other
hand, produces a uniform sample that satisfies all constraints.

\begin{figure}[ht!]
  \begin{center}
    \begin{tabular}[b]{c}
      \includegraphics[width=0.15\textwidth]{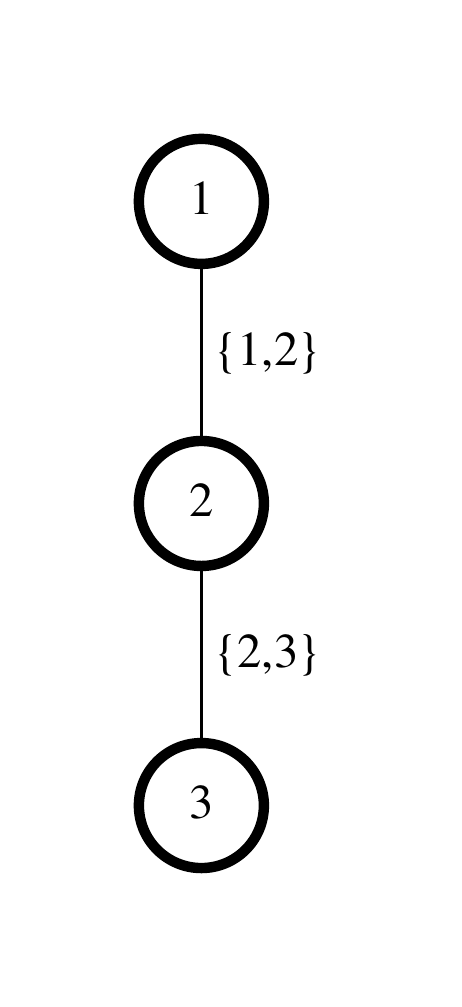}\\
      \small (a)
    \end{tabular}
    \begin{tabular}[b]{c}
      \includegraphics[width=0.39\textwidth]{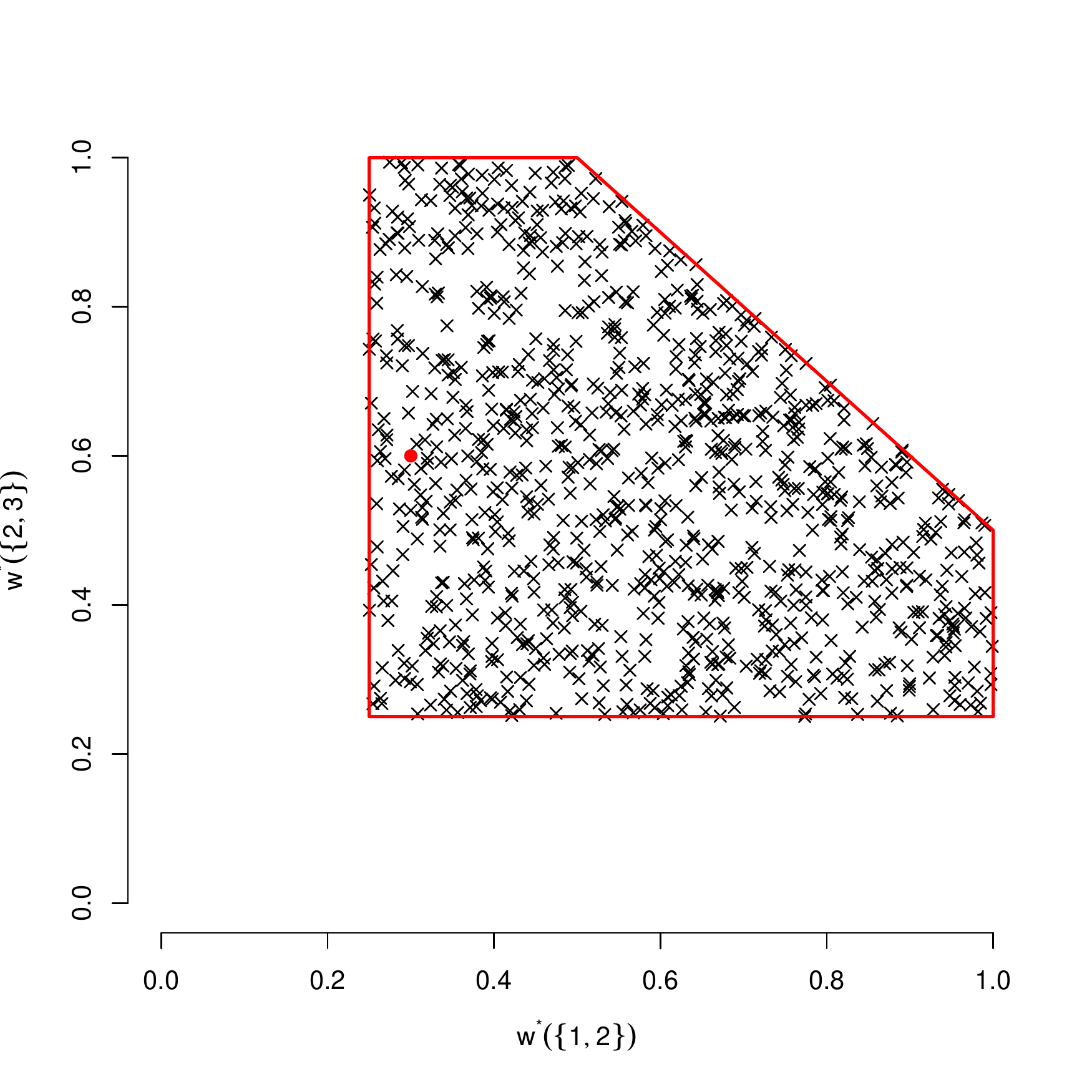}\\
      \small (b)
    \end{tabular}
    \begin{tabular}[b]{c}
      \includegraphics[width=0.39\textwidth]{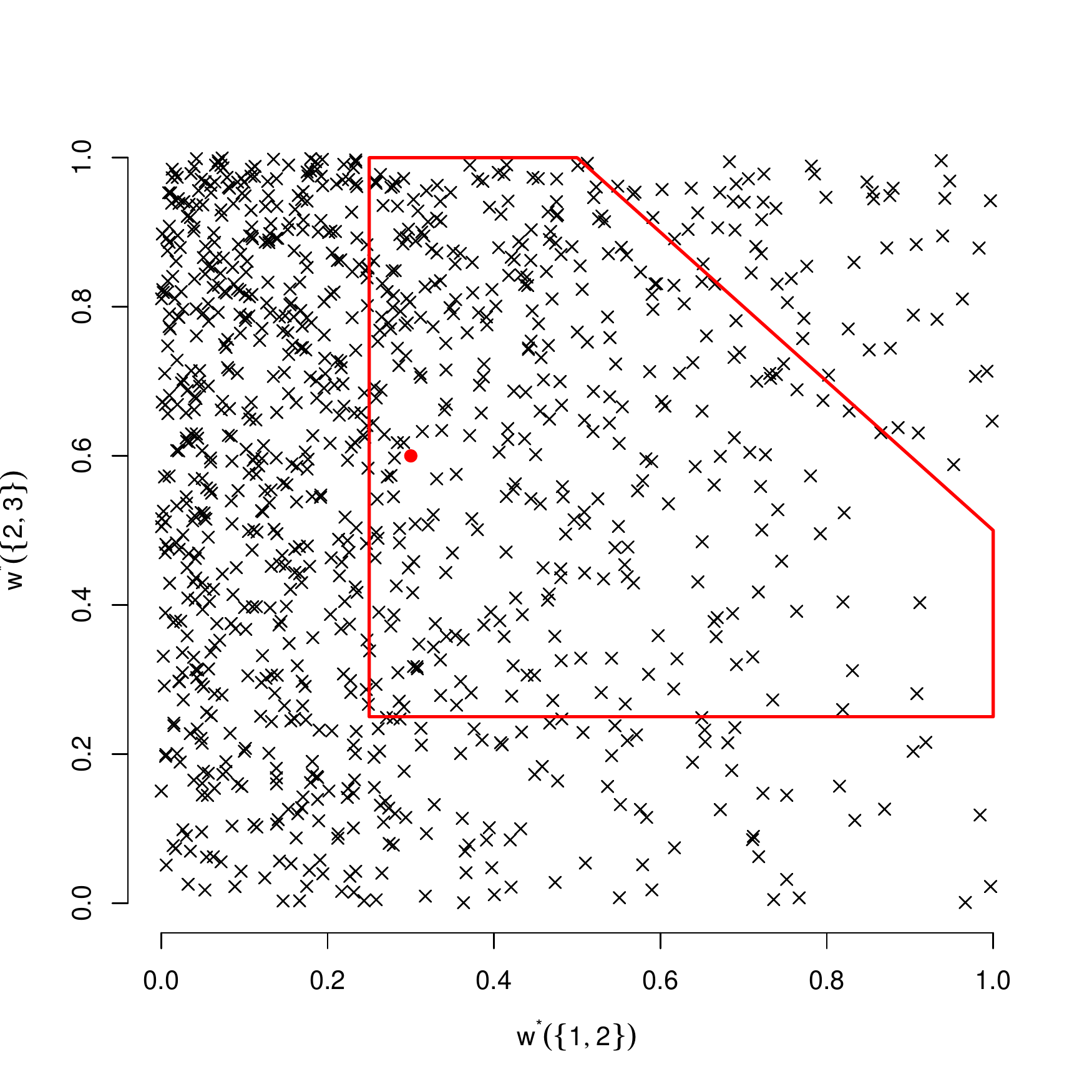}\\
      \small (c)
    \end{tabular}
    \caption{An example illustrating the difference
      between a uniform distribution of weights (the CycleSampler
      method) and the maximum entropy model.  (a) A graph with three
      vertices $\{1,2,3\}$ and two edges with observed weights
      $w(\{1,2\})=0.3$ and $w(\{2,3\})=0.6$, respectively. The node
      weights are $W(1)=w(\{1,2\})=0.3$,
      $W(2)=w(\{1,2\})+w(\{2,3\})=0.9$, and $W(3)=w(\{2,3\})=0.6$. We
      further assume that edge weights are constrained to
      $w^*(e)\in[0,1]$ and node weights to $W^*(v)\in[0.25,1.5]$.  (b)
      Edge weights $w^*$ sampled from the uniform model introduced in
      this paper. The area bounded by the red polygon encloses the set
      of feasible edge weights subject to the node weight
      constraints. The distribution of sampled edge weights is
      uniform, and both edge and node weights satisfy the
      constraints. The red dot shows the observed (original)
      weights. Note that while we use the observed weights as the
      initial point of our Markov chain, the resulting distribution is
      asymptotically identical for all choices of the initial point as
      long as it satisfies the constraints.  (c) Edge weights $w^*$
      from the maximum entropy distribution over edge weights
      $p(w(\{1,2\}),w(\{2,3\}))$ such that the expected node weights
      match the observed ones, i.e., $E_p[W^*(1)]=0.3$,
      $E_p[W^*(2)]=0.9$, and $E_p[W^*(3)]=0.6$. In the maximum entropy
      model there is no obvious way to restrict both edge and node
      weights to a strict interval simultaneously, and most of the
      sampled weights are indeed outside the area bounded by the red
      polygon. Furthermore, the distribution of weights is generally
      not uniform.}
    \label{fig:toy3}
  \end{center}
\end{figure}

\textbf{Summary of contributions.}  In this paper we present the
\emph{CycleSampler} method, which is a {\em structure preserving
  sampling method for edge weights that explicitly maintains node
  weights within a given interval}. This interval can be set to have
zero width, in which case the node weights are maintained exactly in
the sampled networks.  The approach can be viewed as a generalisation
of the property-preserving MCMC algorithm described in
\cite{ansmann2011constrained} to the structure-preserving case.
However, the requirement to not introduce new edges or remove existing
ones presents some nontrivial algorithmic challenges.

In short, our approach samples uniformly from the {\em null space} of
the given network's {\em incidence matrix} (a binary matrix where
vertices are rows and edges are columns), which requires constructing
a basis for this null space.  It is crucial that the basis is {\em
  sparse}, since the null space may be very high-dimensional (in the
millions), and keeping a dense basis in memory, as found by textbook
methods, is infeasible for larger networks.  The problem of finding a
sparse basis for general matrices has been studied previously
\cite{gilbert1987computing,coleman1987null}, and some variants of it
are NP-hard \cite{coleman1986null}.  However, a sparse basis for the
null space of an incidence matrix can be constructed very efficiently
using a spanning tree of the original network \cite{AKBARI2006617}.
This results in an efficient and scalable algorithm that can generate
surrogate networks having millions of edges, while preserving node
weights as described.  We also present an empirical evaluation of the
scalability of our algorithm in a number of real-world cases. We
provide an open source implementation of the CycleSampler method in R,
with time-critical parts implemented in C.

\section{The CycleSampler algorithm}
\label{sec:methods}
\subsection{Problem definition}
\label{sec:definition}
Let $G=(V,E)$ be a graph, where the $m$ vertices are given by $V=[m]$,
where we denote $[m]=\{1,\ldots,m\}$, and the edges by $E\subseteq
\cup_{v\in V}{\cup_{v'\in
    V}{\left\{\left\{v,v'\right\}\right\}}}$. The weight of an edge
$e\in E$ is denoted $w(e)\in{\mathbb{R}}$. We assume that there are no
self-loops in the observed graph, i.e., $|e|=2$ for all $e\in E$.

We use the neighbourhood function $n(v)$, where $v\in V$, to represent
the set of edges connected to a vertex $v$, i.e.,
\begin{equation}\label{eq:nv}
  n(v)=\left\{e\in E\mid v\in e\right\}.
\end{equation}
We define the {\em weight of a vertex} $v\in V$ as the sum of the
weights of the edges connected to it, i.e.,
\begin{equation}\label{eq:W}
  W(v)=\sum_{e\in n(v)}{w(e)}.
\end{equation}
In colloquial terms, our task is to obtain a uniform sample of edge
weights $w(e)$, such that the weights of every edge and every vertex
remain within a given interval.  This problem can be formally defined
as follows.
\begin{problem}\label{pro:sampling}
  Given a connected graph $G=(V,E)$ and a set of intervals
  $[a(e),b(e)]$ for each edge $e\in E$ and $[A(v),B(v)]$ for each
  vertex $v\in V$, respectively, such that $a(e)\le w(e)\le b(e)$ and
  $A(v)\le W(v)\le B(v)$, obtain a sample uniformly at random from the
  set of allowed edge weights ${\cal W}^*$, given by
  \begin{equation}\label{eq:allowed}
    {\cal W}^*=\left\{
    w^*:E\mapsto{\mathbb{R}}
    \mid
    \left(
    \forall e\in E\ldotp
    w^*(e)\in[a(e),b(e)]
    \right)
    \wedge
    \left(
    \forall v\in V\ldotp
    \sum_{e\in n(v)}{w^*(e)}\in[A(v),B(v)]
    \right)
    \right\}.
  \end{equation}
\end{problem}
Existing network sampling algorithms cannot be directly applied to
solve Problem~\ref{pro:sampling} for two reasons.  Firstly, we aim to
preserve network structure, i.e., we can only modify weights on
existing edges, not introduce new edges.  Secondly, we allow vertex
weights to vary within given intervals, while other approaches either
aim to maintain them exactly, or only in expectation.

\subsection{Interval constraints on vertices}
\label{sec:vertexconstraints}
In Problem~\ref{pro:sampling} we allow the weight of each vertex $v\in
V$ to vary on the interval $[A(v),B(v)]$. This sampling problem
becomes easier if these interval constraints are replaced with
equality constraints, i.e., a variant of the problem where the vertex
weights are preserved {\em exactly}.

To do this, we define an equivalent graph containing self-loops (edges
of type $e_v=\{v\}$) where the vertex weights are fixed, i.e.,
$W(v)=A(v)=B(v)$. The variability in the vertex weights is absorbed in
the self-loops. This graph with self-loops can be constructed using
the transformation shown in Fig.~\ref{fig:loopequivalence}. Using this
scheme, any graph with interval constraints on vertex weights and no
self-loops can be transformed to a graph with equality constraints on
vertex weights and self-loops, and vice versa.

\begin{figure}[ht!]
  \begin{center}
    \begin{tabular}[b]{c}
      \includegraphics[width=0.35\textwidth]{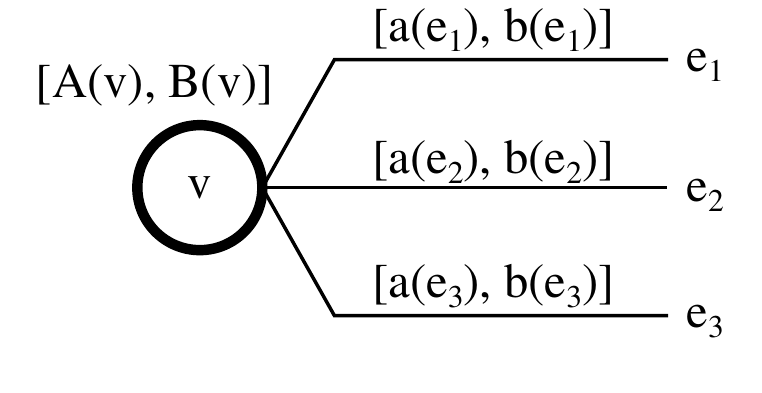}\\
      \small (a)
    \end{tabular}
    \hspace*{0.15\textwidth}
    \begin{tabular}[b]{c}
      \includegraphics[width=0.35\textwidth]{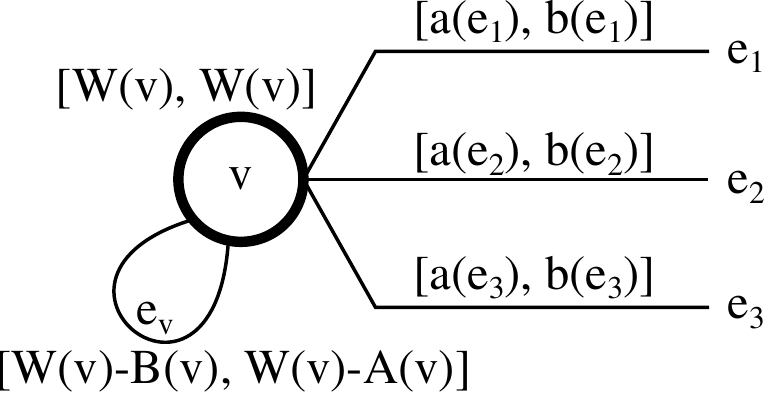}\\
      \small (b)
    \end{tabular}
    \caption{In the graph in (a) there are no self-loops and the
      vertex $v\in V$ has a range of allowed weights $[A(v),B(v)]$. In
      the graph in (b) the weight of vertex $v$ is fixed to $W(v)$ and
      there is a self-loop $e_v=\{v\}$ with a range of allowed weights
      given by $[W(v)-B(v),W(v)-A(v)]$, where
      $W(v)=w(e_1)+w(e_2)+w(e_3)$. These two graphs are equivalent in
      the sense that the allowed ranges of the weights of the edges
      $e_1$, $e_2$, and $e_3$ as the vertex weight (without the
      self-loop) $w(e_1)+w(e_2)+w(e_3)$ are the same for both
      graphs. It follows that a set of allowed weights ${\cal W}^*$
      given by Eq.~\eqref{eq:allowed} defined without self-loops and
      ranges for vertex weights, can be equivalently defined by graphs
      with self loops and fixed vertex weights.}
    \label{fig:loopequivalence}
  \end{center}
\end{figure}

We can now rewrite Eq.~\eqref{eq:allowed} as follows:
\begin{equation}
  \label{eq:allowed2}
        {\cal W}^*=\left\{
        w^*:E\mapsto{\mathbb{R}}
        \mid
        \left(
        \forall e\in E\ldotp
        w^*(e)\in[a(e),b(e)]
        \right)
        \wedge
        \left(
        \forall v\in V\ldotp
        \sum_{e\in n(v)}{w^*(e)}=W(v)
        \right)
        \right\}.
\end{equation}
The problem hence becomes to obtain a uniform sample from the set
${\cal W}^*$ of Eq.~\eqref{eq:allowed2}.  Note that in
Eq.~\eqref{eq:allowed2} the set of edges $E$ now contains self-loops
on those vertices that originally had interval constraints.  In the
following we therefore assume ---without loss of generality--- that
there are only equality constraints on vertex weights, i.e.,
$A(v)=B(v)$ for all $v\in V$.

\subsection{High level approach}
We continue by outlining a sketch of our solution to
Problem~\ref{pro:sampling}. At a high level our algorithm is a Markov
chain Monte Carlo method similar to the algorithm proposed in
\cite{ansmann2011constrained}. It starts from the observed set of edge
weights, introduces a small perturbation to a few of these at every
step, and runs until convergence. Let $C$ denote a \emph{collection of
  subsets of $E$}, i.e., every $E^\prime \in C$ is some (small) set of
edges.  The algorithm in \cite{ansmann2011constrained} as well as ours
can be sketched within a common framework as follows:
\begin{enumerate}
\item Initially, let the current state be the observed set of edge
  weights.
\item Select some $E^\prime \in C$ uniformly at random.
\item Perturb the weights of every edge in $E^\prime$ so that all
  constraints remain satisfied. (Exactly how this is done is described
  in detail below.)
\item Repeat steps 2--3 until convergence.
\end{enumerate}
The main difference between the method in
\cite{ansmann2011constrained} and the method presented here concerns
what the collection $C$ contains. It is crucial to make sure that $C$
is constructed such that the resulting Markov chain indeed converges
to a uniform distribution over ${\cal W}^*$. In other words, every
point in ${\cal W}^*$ must be reachable from every other point in
${\cal W}^*$ by a sequence of steps defined by $C$.

In the algorithm of \cite{ansmann2011constrained}, $C$ contains all
possible cycles of length four. Since in \cite{ansmann2011constrained} 
the underlying graph is assumed
to be a clique, there are plenty of
such cycles, and it can be shown that these are enough for the Markov
chain to reach a uniform distribution. Also, it is fairly easy to see
that in cycles of length four the edge weights can always be adjusted
in a simple manner so that all constraints remain satisfied. However,
in our case finding a suitable $C$ is \emph{complicated by the requirement
of not introducing new edges}. Simply choosing all cycles of length
four from $G$ is not enough. In the remainder of the paper we discuss
our main technical contribution:
\emph{an approach for constructing $C$ in general undirected graphs (not
only cliques) so that the resulting Markov chain converges to the
uniform distribution over ${\cal W}^*$.}
In addition, we show how a simple
transformation of the input graph allows us to \emph{extend the approach
also to directed graphs}.

\subsection{Solution for undirected graphs}
We first describe the solution for general undirected graphs.  We
assume for simplicity of discussion and without loss of generality
that the input graph $G$ consists of a single connected component.
(In general, we can independently sample each of the connected
components of the graph in the sampling process introduced later).
The problem of constructing a suitable collection $C$ of edges becomes
easier if we view our sampling problem in terms of systems of linear
equations.  This way we can express our problem using known concepts
from linear algebra.

As a first step, define the {\em incidence matrix} ${\bf
  A}\in\left\{0,1\right\}^{|V|\times|E|}$ of the graph $G$ in the
usual manner as ${\bf A}_{ve}=2I\left(v\in e\right)/|e|$.  Here $v\in
V$ and $e\in E$ and $I(\Box)$ is an indicator function which equals
unity if $\Box$ is true and is zero otherwise. Also, let ${\bf W} \in
\mathbb{R}^{|V|}$ denote the vector of observed vertex weights defined
by ${\bf W}_v=W(v)$ for all $v\in V$, and denote by ${\bf w}^* \in
\mathbb{R}^{|E|}$ the vector of edge weights. Given these, sampling
uniformly from ${\cal W}^*$ of Eq.~\eqref{eq:allowed2} is equivalent
to the problem of sampling uniformly from the set
\begin{equation}
  {\cal S}=\left\{ {\bf w}^*\in{\mathbb{R}}^{|E|}\mid {\bf A}{\bf w}^*={\bf
    W}\wedge\forall e\in E\ldotp {\bf w}^*_e\in\left[a(e),b(e)\right]
  \right\}.
\end{equation}
By our assumption the original observed weight vector ${\bf w}$ is in
${\cal W}^*$. It follows that ${\bf A}{\bf w}={\bf W}$ and therefore
${\bf w}\in{\cal S}$.

For the moment, let us focus only on the underdetermined linear system
${\bf A}{\bf w}^*={\bf W}$.  A simple known property of such systems
is that their solution space can be expressed as a sum of a single
known solution such as ${\bf w}\in{\cal S}$ and {\em any} vector ${\bf
  x}$ from the {\em null space} of ${\bf A}$. The null space of ${\bf
  A}$, denoted by ${\rm Null}({\bf A})$, is defined by the set ${\rm
  Null}({\bf A})=\left\{{\bf x}\in{\mathbb{R}}^{|E|}\mid {\bf A}{\bf
  x}={\bf 0}\right\}$.  It is easy to see that ${\bf A}{\bf w}^*={\bf
  W}$, where ${\bf w}^*={\bf w}+{\bf x}$, for any ${\bf x}\in{\rm
  Null}({\bf A})$.  Because of the constraints on edge weights, we
cannot simply use {\em any} ${\bf x}\in{\rm Null}({\bf A})$.  Instead,
${\bf x}$ must come from a {\em convex subset} of ${\rm Null}({\bf
  A})$.  Therefore, the problem of sampling uniformly from ${\cal S}$
is equivalent to the problem of sampling uniformly from said convex
subset of ${\rm Null}({\bf A})$.

This is a known problem and could be solved using textbook methods,
such as those described, e.g., in \cite{van2009xsample}.  Those
approaches, however, must usually compute a {\em basis} for ${\rm
  Null}({\bf A})$, which in general is a dense matrix of size $|E|
\times {\rm dim}({\rm Null}({\bf A}))$, where the cardinality of the
null space of ${\bf A}$ is in the same order of magnitude as $|E|$.
While this is not a problem as long as the incidence matrix ${\bf A}$
is fairly small, even storing such a matrix clearly becomes infeasible
for very large networks.  However, since ${\bf A}$ is the incidence
matrix of a network, a {\em sparse basis} is easily constructed by
combining {\em cycles} of $G$, as shown, e.g., in
\cite{AKBARI2006617}.

In short, this works as follows. We first find a {\em spanning tree}
$T$ of $G$. Every edge that does not belong to $T$ clearly induces a
cycle when combined with edges in $T$. Given $T$, every such
even-length cycle is directly an element of the basis, while
odd-length cycles are paired together to form elements of the basis
until it is complete. Details are given in
Section~\ref{sec:background} below. Note that such a sparse basis is
easily represented by a collection $C$ of subsets of edges, together
with appropriate weights for every edge.

\subsection{Directed graphs}
Next we show that the above discussed solution to the sampling problem
for undirected graphs (Problem~\ref{pro:sampling}) can also be applied
to directed graphs, by first transforming the directed graph into an
equivalent undirected graph. The algorithm proposed in this paper can
therefore directly be used to sample weights for both directed and
undirected graphs.

Let $G_D=(V_D,E_D)$ be a directed graph, where the $m_D$ vertices are
given by $V_D=[m_D]$ and the edges by $E_D\subseteq V\times V$. The
weight of the directed edge $e\in E_D$ is denoted by
$w_D(e)\in{\mathbb{R}}$. We define the {\em outgoing} edges of vertex
$v\in V_D$ as
\begin{equation}
  n_o(v)=\left\{(v',v'')\in E_D\mid v'=v\right\},
\end{equation}
and the {\em incoming} edges as
\begin{equation}
  n_i(v)=\left\{(v',v'')\in E_D\mid v''=v\right\}.
\end{equation}
The outgoing weight of a  vertex is given by
\begin{equation}
  W_o(v)=\sum_{e\in n_o(v)}{w_D(e)},
\end{equation}
and the incoming weight by
\begin{equation}
  W_i(v)=\sum_{e\in n_i(v)}{w_D(e)},
\end{equation}

We are now ready to define the sampling problem for directed graphs.
\begin{problem}\label{pro:samplingd}
  Given a connected directed graph $G_D=(V_D,E_D)$ and a set of
  intervals $[a_D(e),b_D(e)]$ for each edge $e\in E_D$ and
  $[A_o(v),B_o(v)]$ and $[A_i(v),B_i(v)]$ for each vertex $v\in V_D$,
  respectively, such that $a_D(e)\le w_D(e)\le b_D(e)$, $A_o(v)\le
  W_o(v)\le B_o(v)$, and $A_i(v)\le W_i(v)\le B_i(v)$, obtain a sample
  uniformly at random from the set of allowed weights ${\cal W}^*_D$,
  given by
  \begin{equation}\label{eq:allowedd}
    \begin{array}{ll}
      {\cal W}^*_D=\left\{
      w^*_D:E_D\mapsto{\mathbb{R}}
      \mid\right.&
      \left(
      \forall e\in E_D\ldotp
      w^*_D(e)\in[a_D(e),b_D(e)]
      \right)
      \wedge\\&
      \left(
      \forall v\in V_D\ldotp
      \sum_{e\in n_o(v)}{w^*_D(e)}\in[A_o(v),B_o(v)]
      \right)
      \wedge\\&\left.
      \left(
      \forall v\in V_D\ldotp
      \sum_{e\in n_i(v)}{w^*_D(e)}\in[A_i(v),B_i(v)]
      \right)
      \right\}.
    \end{array}
  \end{equation}
\end{problem}

We can solve Problem \ref{pro:samplingd} by the algorithm used to
solve Problem \ref{pro:sampling} by noticing that a directed graph can
be easily transformed to an equivalent undirected graph, as stated by
the following theorem.
\begin{theorem}\label{thm:dir}
  The set of allowed weights ${\cal W}_D^*$ of Eq.~\eqref{eq:allowedd}
  for a directed graph $G_D$ is equivalent to the set of allowed
  weights ${\cal W}^*$ of Eq.~\eqref{eq:allowed} for an undirected
  bipartite graph $G$ when the graph $G$ is defined as follows. The
  graph $G$ has $m=2m_D$ vertices, i.e., $V=[m]$. The set of
  undirected edges $E$ of $G$ is given by $E=\left\{
  \{v',v''+m_D\}\mid (v',v'')\in E_D \right\}$. We define a mapping
  $f: E\mapsto E_D$ as follows, $f(\{v',v''+m_D\})=(v',v'')$ for all
  $(v',v'')\in E_D$. The weight of an edge $e\in E$ is given by
  $w(e)=w_D(f(e))$ and the bounds by $a(e)=a_D(f(e))$ and
  $b(w)=b_D(f(e))$.  The vertices $1,\ldots,m_D$ correspond to
  outgoing weights and the vertices $m_D+1,\ldots,2m_D$ to incoming
  weights as follows,
  \begin{equation}
    W(v)=\left\{
    \begin{array}{lcl}
      W_o(v)&,&v\le m_D\\
      W_i(v-n_D)&,&v>m_D
    \end{array}
    \right.
    ,
  \end{equation}
  with the bounds given by
  \begin{equation}
    A(v)=\left\{
    \begin{array}{lcl}
      A_o(v)&,&v\le m_D\\
      A_i(v-n_D)&,&v>m_D
    \end{array}
    \right.
    ,
  \end{equation}
  and
  \begin{equation}
    B(v)=\left\{
    \begin{array}{lcl}
      B_o(v)&,&v\le m_D\\
      B_i(v-n_D)&,&v>m_D
    \end{array}
    \right.
    .
  \end{equation}
  Now, if $w^*$ is a uniform sample from ${\cal W}^*$ we can obtain a
  uniform sample $w^*_D$ from ${\cal W}^*_D$ in a straightforward way by
  setting $w^*_D(f(e))\leftarrow w^*(e)$ for all $e\in E$.
\end{theorem}
\begin{proof}
  The proof follows directly from the definitions.
\end{proof}

\subsection{Details of the algorithm}
\label{sec:background}
In this Section we give a detailed proof of the proposed algorithm, an
overview of which is given above.

\vspace*{2em}
\noindent \textbf{The CycleSampler Algorithm}\\
Assume ${\bf y}_i$ where $i\in [l]$ (where $l$ is
no smaller than the dimension of the null space ${\rm Null}({\bf A})$)
spans the null space ${\rm Null}({\bf A})$, i.e., any null space
vector ${\bf x}\in{\rm Null}({\bf A})$ can be formed as a linear
combination of vectors in ${\bf y}_i$. Given this basis, we can obtain
samples as follows:
\begin{enumerate}
\item First, transform the graph to a form where the vertex weights
  are fixed and the variability in them is described by self-loops, as
  in Fig. \ref{fig:loopequivalence}.
\item Initially, let the current state be the observed set of edge
  weights, ${\bf w}^*\leftarrow{\bf w}$, with the weight of
  self-loops initially set to zero.
\item Pick a vector $i\in [l]$ at random and let $[a,b]$ be the
  largest range of allowed values such that ${\bf w}^*+\alpha{\bf
    y}_i$ where $\alpha\in[a,b]$ stays within ${\cal W}^*$. Sample
  $\alpha$ uniformly at random from $[a,b]$.
\item Update ${\bf w}^*\leftarrow{\bf w}^*+\alpha{\bf y}_i$ and
  repeat from step 3 above.
\end{enumerate}
Note that because ${\cal W}^*$ is a simple convex space---an
$|E|$-dimensional rectangle---we can find $[a,b]$ for a given ${\bf
  y}_i$ efficiently by a simple loop over the non-zero dimensions of
${\bf y}_i$. The updates at step 4 form a Markov chain of edge weight
vectors ${\bf w}^*$.
\vspace*{2em}

\begin{theorem}
  \label{theorem:uniformsample}
  The \texttt{CycleSampler} algorithm asymptotically
provides (after a sufficient number of iterations) a uniform
sample from the set ${\cal W}^*$.
\end{theorem}
\begin{proof}
This follows from the facts that (i) because ${\bf y}_i$ span the null
space all of the points of the null space can be reached with a
non-vanishing probability and that (ii) the transition probability
from state ${\bf w}^*\in{\cal S}$ to state ${{\bf w}^*}'\in{\cal S}$
is equal to the transition probability from state ${{\bf w}^*}'$ to
state ${\bf w}^*$, i.e., the Markov chain satisfies the detailed
balance condition for a uniform distribution.
\end{proof}
It remains to find a complete basis ${\bf y}_i$ of the null space
${\rm Null}({\bf A})$. This can be done using a spanning tree of the
connected graph $G$. In the actual implementation we use a spanning
tree found by a standard breadth-first search, but any instance of a
spanning tree can be used. We choose one of the vertices (it does not
matter which) as the root vertex $v_{root}\in V$ of the tree. We
denote by $E_s\subseteq E$ the $|V|-1$ edges that appear in the
spanning tree and by $F=E\setminus E_s$ the remaining $|E|-|V|+1$
edges that do not appear in the spanning tree. To make the derivation
easier to follow we provide an example of the introduced concepts
later in Sec. \ref{sec:example} which can be read in parallel with the
derivation below.

We further denote by $E(v)\subseteq E_s$ the set of edges in the
spanning tree between vertex $v\in V$ and the root vertex. For the
root vertex $v_{root}$ we have $E(v_{root})=\emptyset$. We define the
depth ${\rm depth}(e)$, where $e\in E_s$, of the edge $e$ in the
spanning tree to be the number of edges between $e$ and the root
vertex, the edges adjacent to the root vertex having a depth of
zero. We further define the depth of vertex $v$ by its distance from
the root vertex, i.e., ${\rm depth}(v)=|E(v)|$.

We define a {\em cycle} ${\bf c}(v,v')\in{\mathbb{R}}^{|E|}$ for
each edge $\{v,v'\}\in F$ by
\begin{equation}\label{eq:cycle}
  {\bf c}(v,v')={\bf n}(\{v,v'\})
  +\sum_{e\in E(v)}{(-1)^{{\rm depth}(v)+{\rm depth}(e)}{\bf n}(e)}
  +\sum_{e'\in E(v')}{(-1)^{{\rm depth}(v')+{\rm depth}(e')}{\bf n}(e')},
\end{equation}
where ${\bf n}(e)\in \{0,1\}^{|E|}$ is a 0-1 vector defined by ${\bf
  n}(e)_{e'}=I(e=e')$. We use the shorthand-notation to denote ${\bf
  c}(e)={\bf c}(v,v')$ where $e=\{v,v'\}\in F$.

We further split the edges not in the spanning tree $F$ into {\em
  clean edges},
\begin{equation}
  F_c=\left\{\{v,v'\}\in F \mid {\rm depth}(v)+{\rm depth}(v') {\rm ~is~odd}\right\},
\end{equation}
and {\em dirty edges},
\begin{equation}
  F_d=\left\{ \{v,v'\}\in F \mid v=v' \vee \left({\rm depth}(v)+{\rm
    depth}(v') {\rm ~is~even}\right)\right\}.
\end{equation}
Notice that the set of clean edges $F_c$ cannot contain self-loops,
but the set of dirty edges $F_d$ may contain self-loops (i.e., $v =
v'$). The non-zero elements of a cycle introduced by clean edges
contain graph loops with an even number of edges, while a cycle
introduced by a dirty edge contains a graph loop with an odd number of
edges. Indeed, by taking a spanning tree of a graph and adding an edge
not in graph we always get a unique {\em graph cycle} that forms the
{\em cycle basis} of the graph. In particular, it is well-known that
the cycle basis of a graph contains only even-length cycles if and
only if the graph is bipartite. Therefore $F_d=\emptyset$ is
equivalent to the statement that the graph $G$ is bipartite.

We construct a basis of the null space as follows.

(i) For each clean edge $e\in F_c$ we define a basis vector by the
respective cycle,
\begin{equation}\label{eq:yo}
  {\bf y}_i={\bf c}(e).
\end{equation}
For an example, see Fig. \ref{fig:basis}a. The non-zero elements of
the basis vector ${\bf y}_i$ form a graph cycle of even length, with
alternating weights of $\pm 1$.

(ii) For each pair of dirty edges $e_1\in F_d$ and $e_2\in F_d$ where
$e_1\ne e_2$ the basis vector is given by a linear combination of two
cycles,
\begin{equation}\label{eq:ye}
  {\bf y}_i={\bf c}(e_1)-(-1)^{{\rm depth}(e_1)+{\rm depth}(e_2)}{\bf c}(e_2).
\end{equation}
For an example, see Fig. \ref{fig:basis}b--d. Again, the non-zero
elements of the basis vector ${\bf y}_i$ form a graph cycle of even
length.

The number of distinct basis vectors, defined by Eqs. \eqref{eq:yo}
and \eqref{eq:ye}, is therefore $|F_c|+|F_d|\left(|F_d|-1\right)/2$.

We show that the basis vectors of Eqs. \eqref{eq:yo} and \eqref{eq:ye}
form a complete basis of the null space ${\rm Null}({\bf A})$ by first
proving the following three lemmas.

\begin{lemma}\label{lemma:yo}
  The vectors defined by Eq. \eqref{eq:yo} are in the null space ${\rm
    Null}({\bf A})$ and they span an $|F_c|$ dimensional space.
\end{lemma}
\begin{proof}
  A vector ${\bf y}_i$ defined by Eq.~\eqref{eq:yo} is in the null space
  ${\rm Null}({\bf A})$, because the equation ${\bf A}{\bf y}_i={\bf 0}$ is
  satisfied for all ${\bf y}_i$.

  The vectors ${\bf y}_i$ are clearly linearly independent, because
  each of the vectors contain a unique non-zero dimension given by an
  edge $e\in F_c$ which is zero in all other vectors. Therefore, the
  $|F_c|$ vectors span an $|F_c|$ dimensional space.
\end{proof}

\begin{lemma}\label{lemma:ye}
  If there are dirty edges, i.e., $F_d\ne\emptyset$, the vectors
  defined by Eq.~\eqref{eq:ye} are in the null space ${\rm Null}({\bf
    A})$, they span an $|F_d|-1$ dimensional space, and they are
  independent of the vectors defined by Eq.~\eqref{eq:yo}.
\end{lemma}
\begin{proof}
  First, assume that there are at least two even edges, because
  otherwise there are no pairs and hence no vectors defined by
  Eq.~\eqref{eq:ye}. A vector defined by Eq.~\eqref{eq:ye} is in the
  null space, because ${\bf A}{\bf y}_i={\bf 0}$ is satisfied.

  The vectors span an $|F_d|-1$ dimensional subspace. This can be seen
  by first arranging the vectors in $F_d$ in an arbitrary order,
  numbered by $1,\ldots,|F_d|$ and by forming $|F_d|-1$ pairs by
  combining the $i$th and $(i+1)$th vectors into pairs, respectively,
  where $i\in[|F_d|-1]$. A pair of $i$th and $(i+1)$th vectors is
  independent of the previous pairs, because it contains a non-zero
  value for the edge $e\in F_d$ related to ${\bf c}_{i+1}$ that has a
  zero value for all of the previous pairs.

  A vector ${\bf y}_i$ defined by Eq. \eqref{eq:ye} is linearly
  independent of any vector defined by \eqref{eq:yo}, because the
  vector contains a non-zero element for two edges in $F_d$ that do
  not occur in any of the vectors defined by Eq.~\eqref{eq:yo}.
\end{proof}

\begin{lemma}\label{lemma:dim}
  The dimensionality of the null space ${\rm Null}({\bf A})$ is $|F_c|$ if
  there are no even edges and $|F_c|+|F_d|-1$ if there are even
  edges.
\end{lemma}
\begin{proof}
  Consider the rank of matrix ${\bf A}$. The rank of the matrix is at
  most the number of its rows, i.e., ${\rm rank}({\bf A})\le |V|$. The
  equality does not hold if and only if there is a non-zero vector
  ${\bf v}\in{\mathbb{R}}^{|V|}$ such that ${\bf v}^T{\bf A}={\bf
    0}$. The vector ${\bf v}$ must satisfy the following
  properties. If a vertex $i\in V$ has a self-loop its weight must be
  zero, i.e., ${\bf v}_i=0$, because the column (edge) that represent
  a self-loop has only one non-zero value. A pair of rows (vertices)
  connected with an edge, $\{i,j\}\in E$, must have opposite signs,
  i.e., ${\bf v}_i=-{\bf v}_j$, otherwise the column (edge) in the
  matrix product ${\bf v}^T{\bf A}$ would be non-zero. Because the
  graph is connected we can construct a vector ${\bf v}$ simply by
  starting from one row (vertex), e.g., $i=1$ and setting ${\bf
    v}_1\leftarrow x$, where $x$ is some number. We can then
  iteratively follow any path in the graph, and assign values for the
  remaining rows in ${\bf v}$, the weights of the items in ${\bf v}$
  are therefore $x$ or $-x$.

  Consider first the case where there are dirty edges, i.e.,
  $F_d\ne\emptyset$, and there is at least one cycle in the graph with
  an odd number of edges. If we follow this cycle it leads to the
  situation where $x=-x$, meaning that ${\bf v}={\bf 0}$ is the only
  viable solution and therefore ${\rm rank}({\bf A})=|V|$. According
  to the rank-nullity therem the rank of the null space is $|E|-{\rm
    rank}({\bf A})=|E|-|V|=|F|-1=|F_c|+|F_d|-1$, which proves the
  lemma for the case $F_e\ne\emptyset$.

  Then consider the case where there are no even edges, i.e.,
  $F_e=\emptyset$. If there are no even edges then all graph cycles
  are of even length. The graph cycle basis has only even cycles if
  and only if the graph is bipartite and all vertices can be labeled,
  e.g., by $-1$ and $+1$ according to this bipartite graph. If the
  vector ${\bf v}$ is constructed according to this labeling then
  ${\bf v}^T{\bf A}={\bf 0}$ is satisfied and we therefore have ${\rm
    rank}({\bf A})\le|V|-1$. Next, we consider the rank of a matrix
  ${\bf A}$ with the first row removed, denoted by ${\bf A}'$. This
  matrix contains (because the graph is connected) at least one column
  which contains only one non-zero entry. The rank of this matrix is
  therefore at least $|V|-1$, from which it follows that the ${\rm
    rank}({\bf A})=|V|-1$. The rank of the null space is then
  according to the rank-nullity theorem $|E|-{\rm rank}({\bf
    A})=|E|-|V|+1=|F|=|F_c|$. This proves the lemma for the case of
  $F_c=\emptyset$.
\end{proof}

The following theorem follows directly from Lemmas \ref{lemma:yo},
\ref{lemma:ye}, and \ref{lemma:dim} above.
\begin{theorem}\label{theorem:complete}
  The basis vectors defined by Eqs. \eqref{eq:yo} and \eqref{eq:ye}
  span the null space ${\rm Null}({\bf A})$ for a connected graph
  $G=(V,E)$. The dimensionality of the null space is $|E|-|V|+1$ if
  the graph $G$ is bipartite and $|E|-|V|$ otherwise.
\end{theorem}

\subsection{Illustrating Example}
\label{sec:example}
In this section we provide a small example graph illustrating the
above discussed concepts. Consider again the network introduced above
in Fig.~\ref{fig:example1:simple}, with six vertices
$V=\{1,2,3,4,5,6\}$ and seven edges. We now add two self-loops to
nodes 1 and 6 in this network, giving the nine edges
\begin{displaymath}
  E=\{\{1\},\{1,2\},\{1,3\},\{1,6\},\{2,3\},\{3,4\},\{4,5\},\{4,6\},\{6\}\}.
\end{displaymath}
Further assume that the edges in the spanning tree of this graph are
given by
\begin{displaymath}
  E_s=\{\{1,3\},\{2,3\},\{3,4\},\{4,5\},\{4,6\}\}
\end{displaymath}
and the
edges not in the spanning tree by
\begin{displaymath}
  F=\{\{1\},\{1,2\},\{1,6\},\{6\}\}.
\end{displaymath}
The root vertex is given by $v_{root}=3$. This graph is shown in
Fig.~\ref{fig:example1} (the root of the spanning tree is marked with
grey) and the corresponding matrix ${\bf A}$ is shown in
Table~\ref{tab:example1}. The cycle vectors are shown in
Fig.~\ref{fig:cycles} and the basis of the null space in
Fig.~\ref{fig:basis}. A different root vertex could be used when
constructing the spanning tree, which would lead to a different set of
vectors that span the null space; however, any choice of spanning tree
or root vertex will span the same null space.

\begin{figure}
  \begin{center}
    \includegraphics[width=0.48\textwidth]{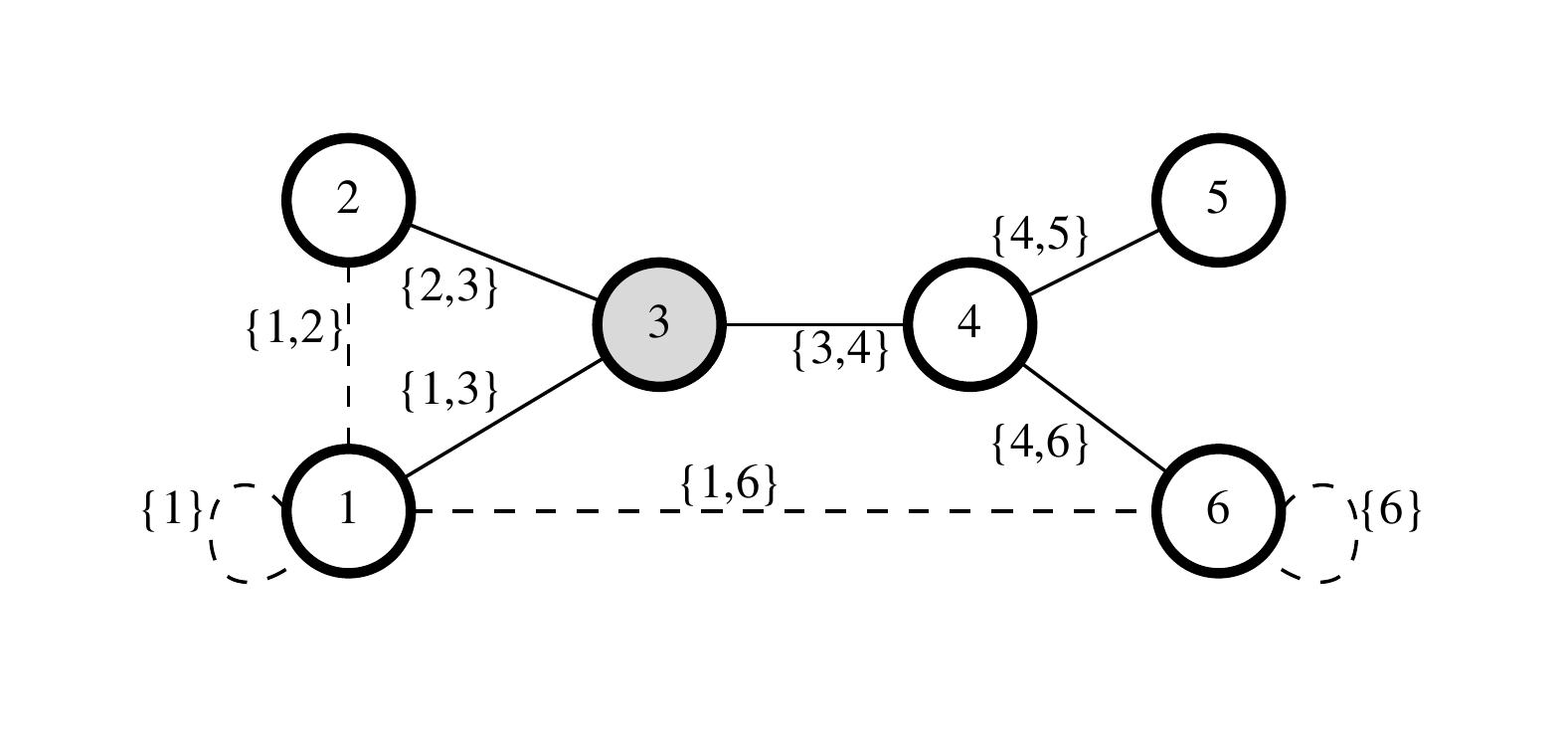}
    \caption{Example graph width six vertices and nine edges. The five
      edges in the spanning tree are shown with solid lines and the
      four edges not in the spanning tree with dashed lines. Vertex
      $3$ is the root vertex of the spanning tree (marked with grey).}
    \label{fig:example1}
  \end{center}
\end{figure}

\begin{table}
  \centering
    \begin{ruledtabular}
      \begin{tabular}{c|ccccccccc}
        ${\bf A}$&$\{1\}$&$\{1,2\}$&$\{1,3\}$&$\{1,6\}$&$\{2,3\}$&$\{3,4\}$&$\{4,5\}$&$\{4,6\}$&$\{6\}$\\\hline
        1&2&1&1&1&0&0&0&0&0 \\
        2&0&1&0&0&1&0&0&0&0 \\
        3&0&0&1&0&1&1&0&0&0 \\
        4&0&0&0&0&0&1&1&1&0 \\
        5&0&0&0&0&0&0&1&0&0 \\
        6&0&0&0&1&0&0&0&1&2 \\\hline\hline
        ${\bf c}_a^T$&$0$&$0$&$-1$&$1$&$0$&$1$&$0$&$-1$&$0$ \\
        ${\bf c}_b^T$&$1$&$0$&$-2$&$0$&$0$&$0$&$0$&$0$&$0$ \\
        ${\bf c}_c^T$&$0$&$1$&$-1$&$0$&$-1$&$0$&$0$&$0$&$0$ \\
        ${\bf c}_d^T$&$0$&$0$&$0$&$0$&$0$&$2$&$0$&$-2$&$1$ \\\hline\hline
        ${\bf y}_a^T$&$0$&$0$&$-1$&$1$&$0$&$1$&$0$&$-1$&$0$ \\
        ${\bf y}_b^T$&$1$&$-1$&$-1$&$0$&$1$&$0$&$0$&$0$&$0$ \\
        ${\bf y}_c^T$&$1$&$0$&$-2$&$0$&$0$&$2$&$0$&$-2$&$1$ \\
        ${\bf y}_d^T$&$0$&$1$&$-1$&$0$&$-1$&$2$&$0$&$-2$&$1$ \\
      \end{tabular}
    \end{ruledtabular}
    \caption{The six uppermost rows show the
      incidence matrix matrix ${\bf A}$ for the graph in
      Fig.~\ref{fig:example1}. The next four rows show the graph
      cycles ${\bf c}_i$ which are also shown graphically in
      Fig.~\ref{fig:cycles}; see Eq. \eqref{eq:cycle} for the
      definition. 
     The four lowermost rows show the basis
      vectors ${\bf y}_i$, also shown graphically in
      Fig.~\ref{fig:basis}; see Eqs. \eqref{eq:yo} and \eqref{eq:ye}
      for the definition. The basis vectors have been constructed
      from the cycles as follows: the basis vector ${\bf y}_a$ by a
      cycle induced by the clean edge, ${\bf y}_a={\bf c}_a$, and the
      remaining three basis vectors ${\bf y}_b,\ldots,{\bf y}_d$ by
      linear combinations of two cycles induced by dirty edges, ${\bf
        y}_b={\bf c}_b-{\bf c}_c$, ${\bf y}_c={\bf c}_b+{\bf c}_d$,
      and ${\bf y}_d={\bf c}_d+{\bf c}_c$. All of the basis vectors
      satisfy ${\bf A}{\bf y}_i={\bf 0}$, and span a 3-dimensional
      null space ${\rm Null}({\bf A})$, as required by Theorem
      \ref{theorem:complete}; the basis vectors constructed using
      dirty edges have a linear dependence given by ${\bf y}_d={\bf
        y}_c-{\bf y}_b$.     \label{tab:example1}}
\end{table}

\begin{figure}
  \begin{center}
    \begin{tabular}[b]{c}
      \includegraphics[width=0.48\textwidth]{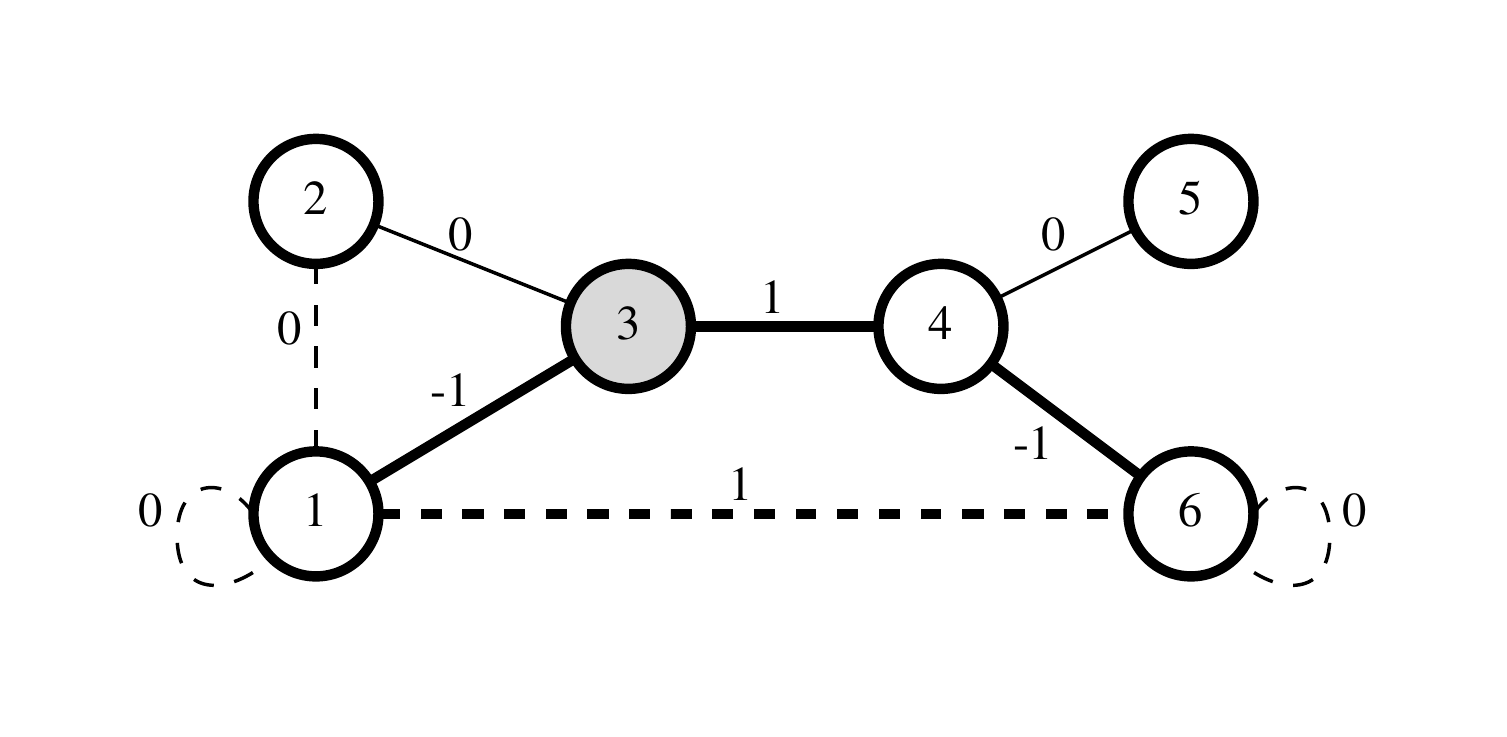}\\
      \small (a)
    \end{tabular}
    \begin{tabular}[b]{c}
      \includegraphics[width=0.48\textwidth]{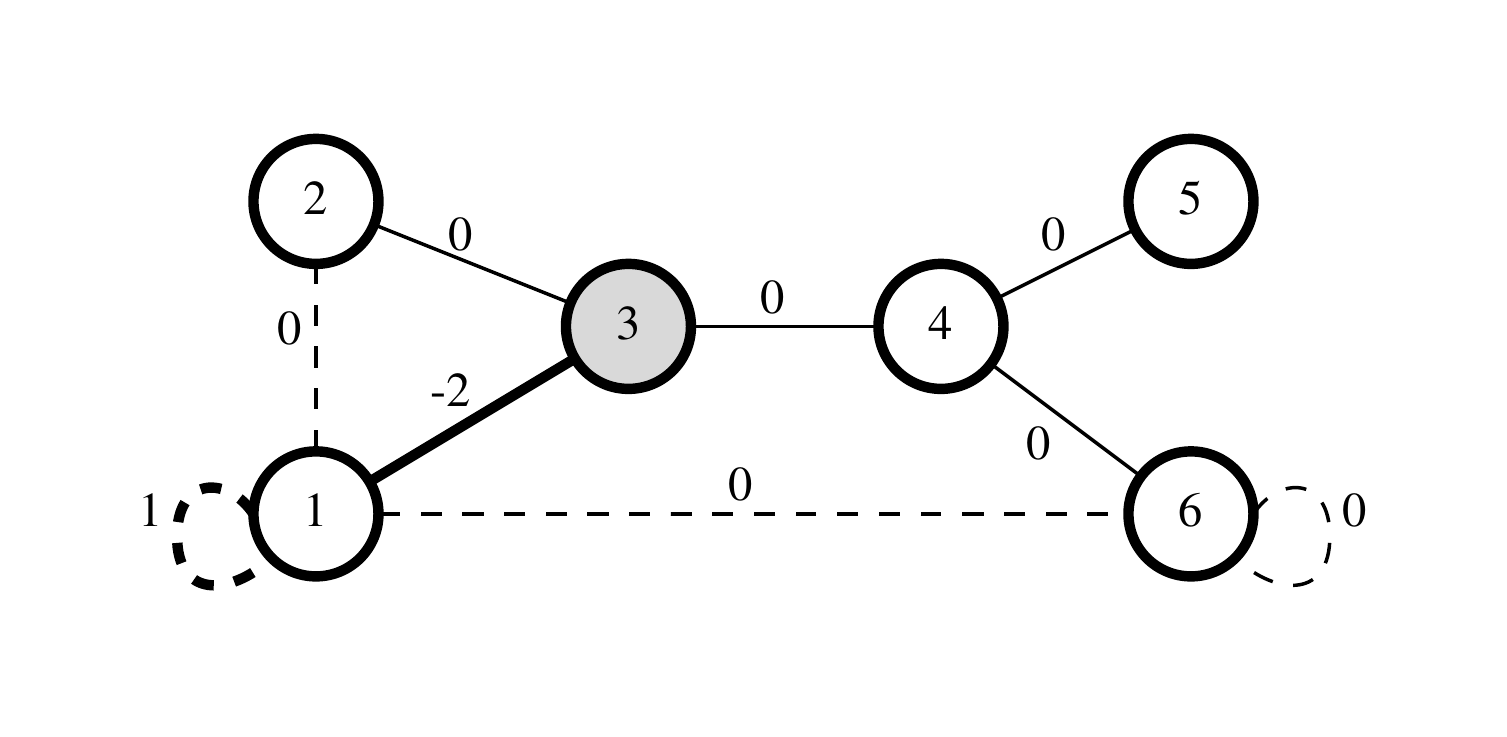}\\
      \small (b)
    \end{tabular}
    \begin{tabular}[b]{c}
      \includegraphics[width=0.48\textwidth]{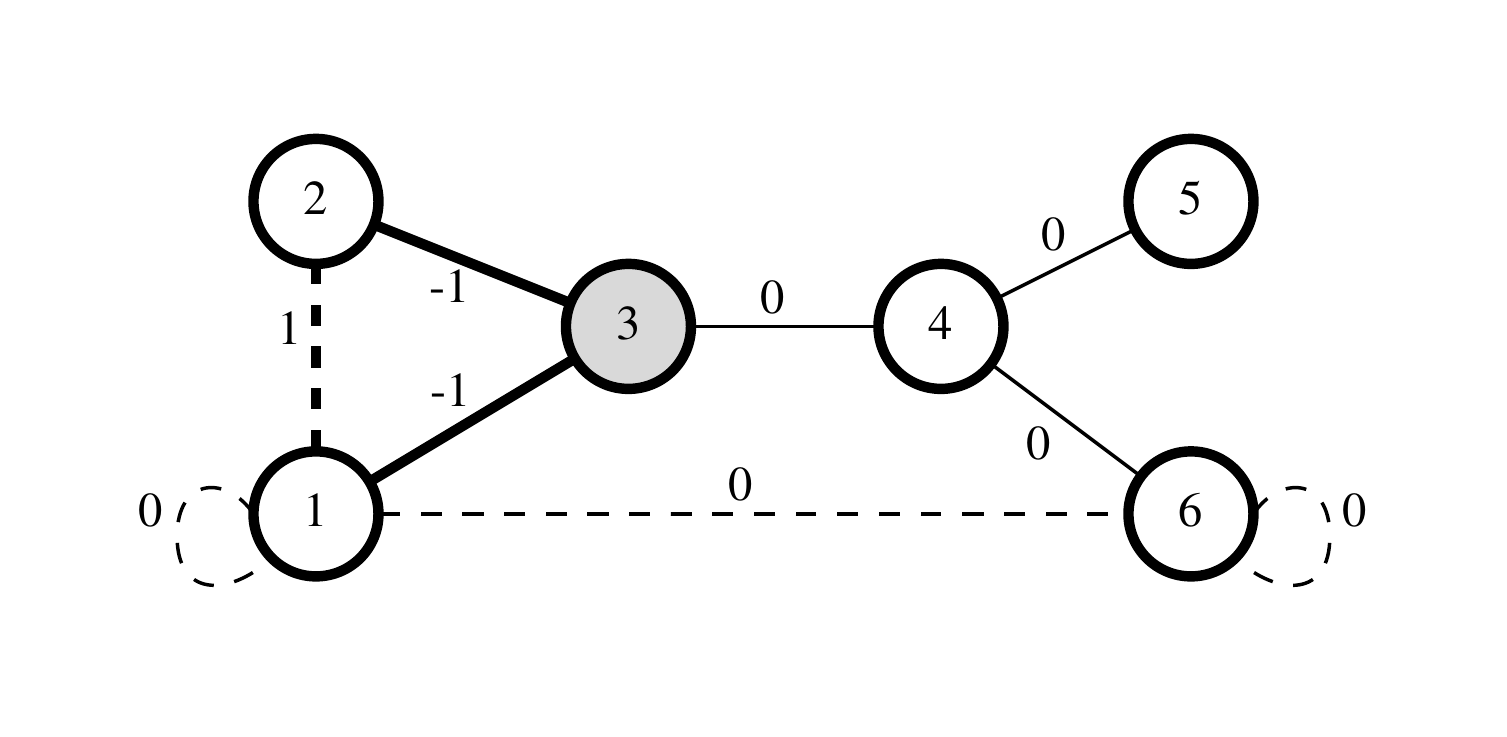}\\
      \small (c)
    \end{tabular}
    \begin{tabular}[b]{c}
      \includegraphics[width=0.48\textwidth]{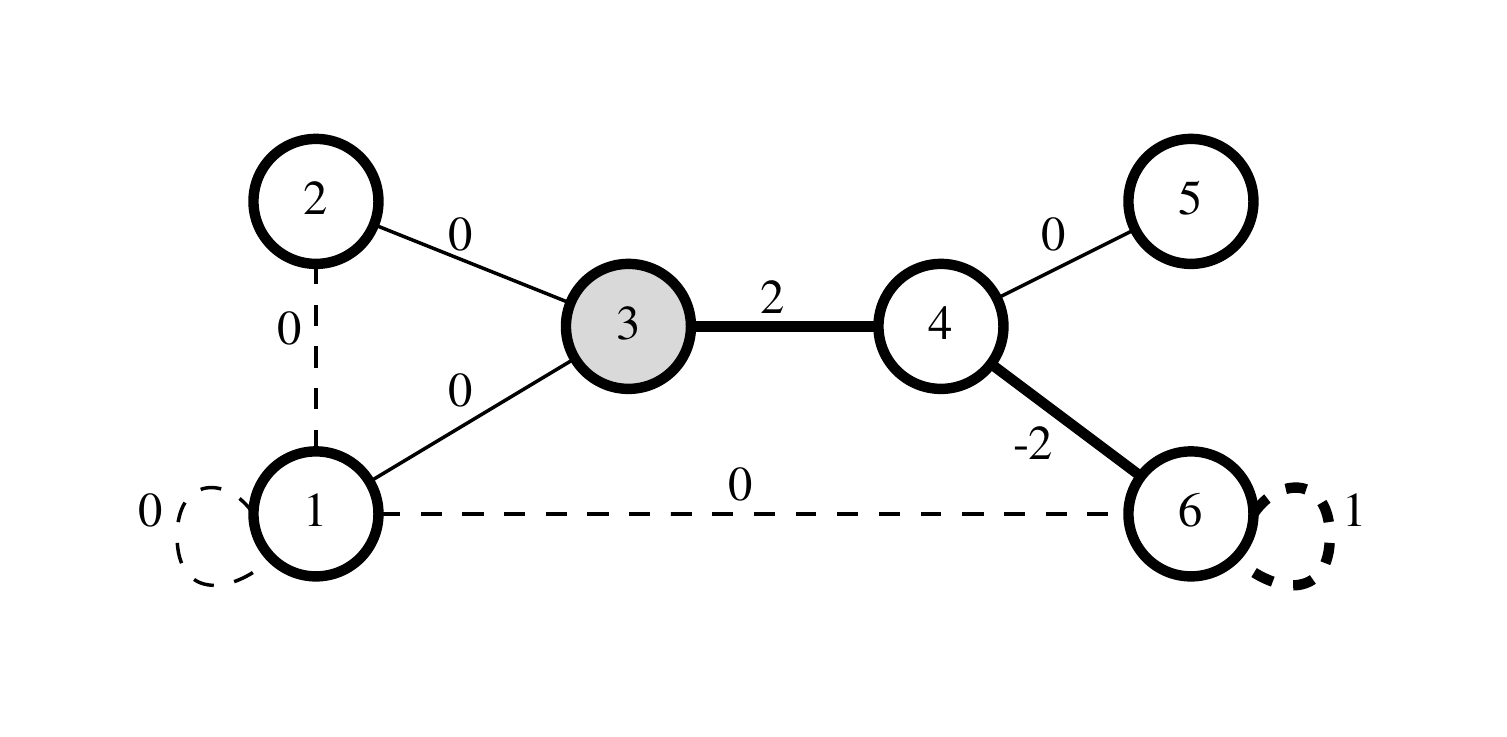}\\
      \small (d)
    \end{tabular}
    \caption{\label{fig:cycles} The cycles of the graph in Fig.
      \ref{fig:example1} and the corresponding values of the
      respective vector ${\bf c}\in{\mathbb{R}}^{|E|}$. There is one
      cycle related to a clean edge (${\bf c}(\{1,6\})$ in (a)) and
      three cycles related to dirty edges: ${\bf c}(\{1\})$ in (b),
      ${\bf c}(\{1,2\})$ in (c) and ${\bf c}(\{6\})$ in (d).}
  \end{center}
\end{figure}

\begin{figure}
  \begin{center}
    \begin{tabular}[b]{c}
      \includegraphics[width=0.48\textwidth]{example_c16}\\
      \small (a)
    \end{tabular}
    \begin{tabular}[b]{c}
      \includegraphics[width=0.48\textwidth]{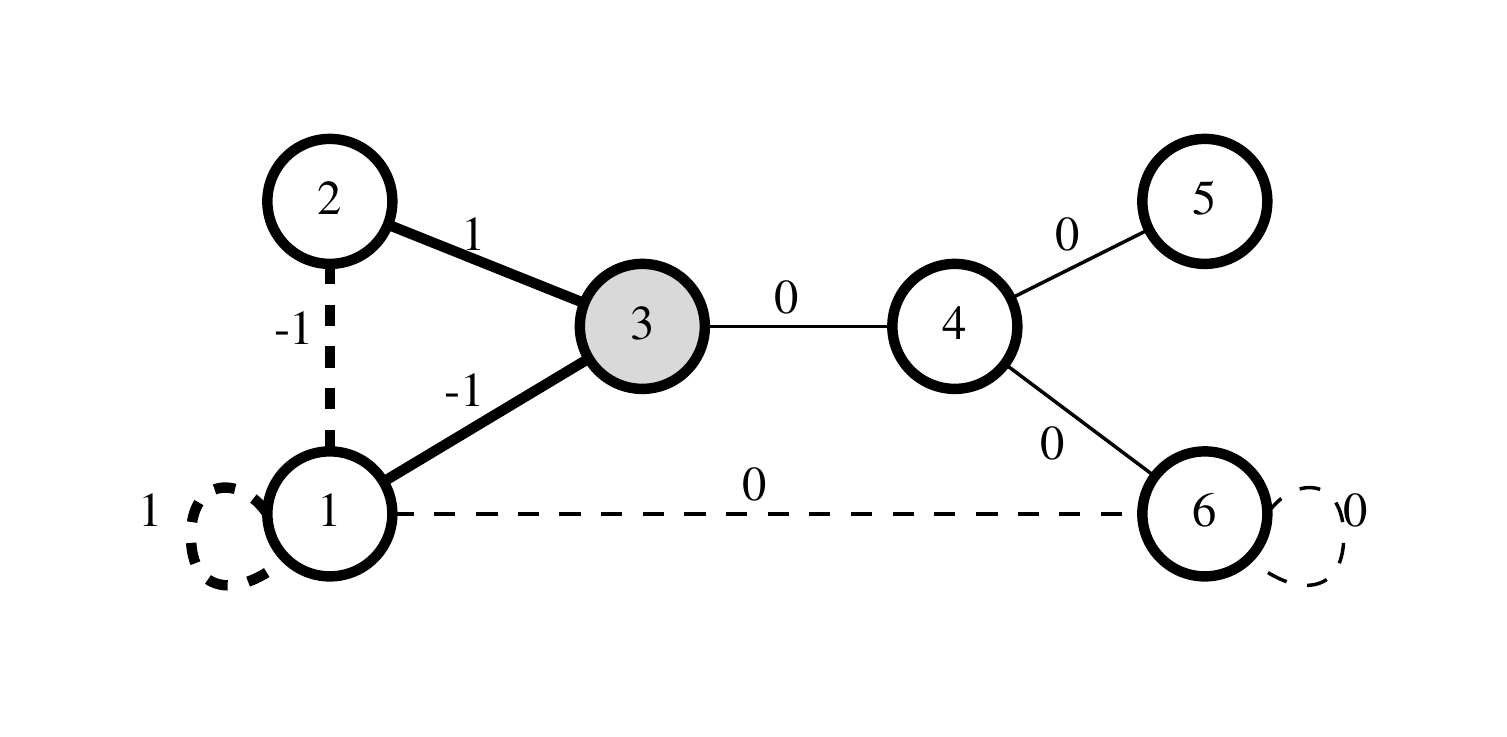}\\
      \small (b)
    \end{tabular}
    \begin{tabular}[b]{c}
      \includegraphics[width=0.48\textwidth]{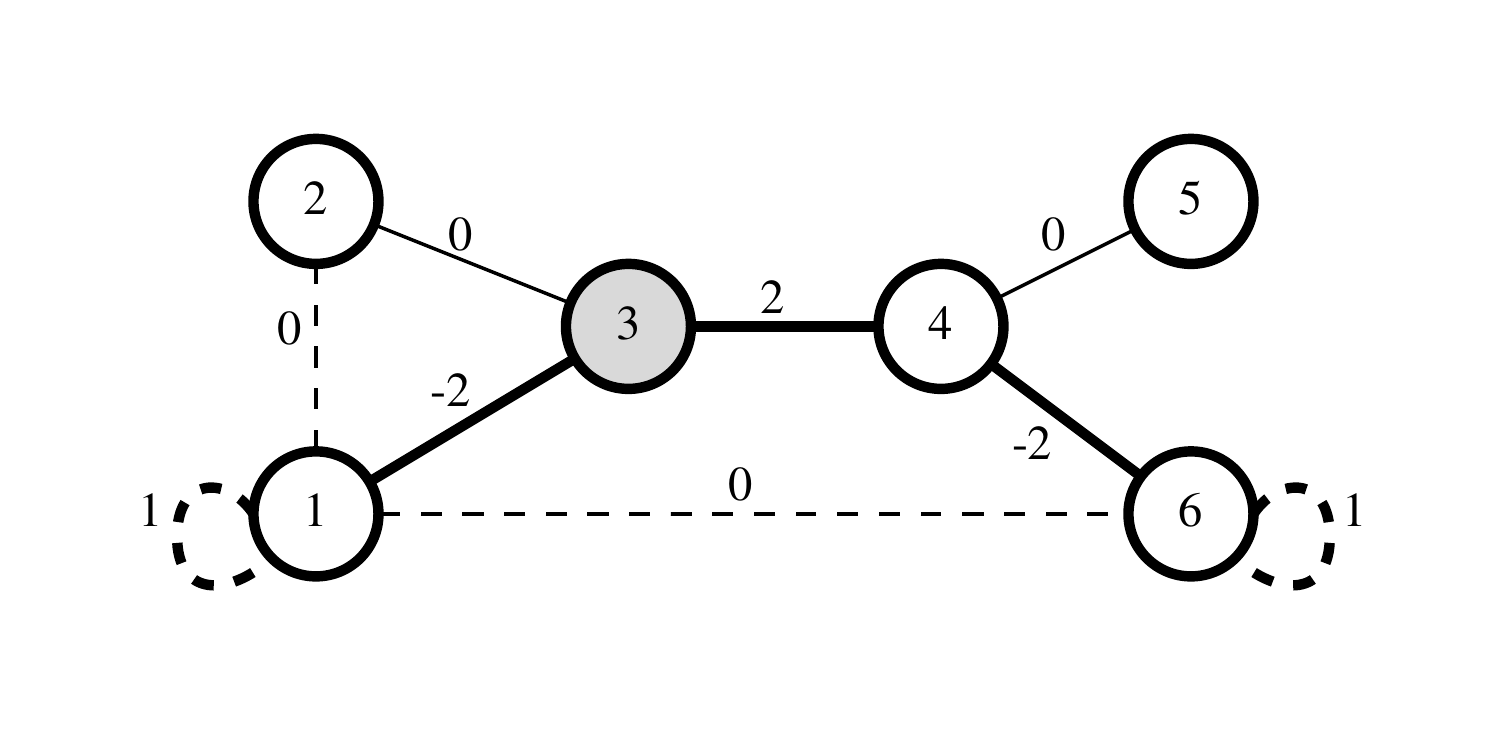}\\
      \small (c)
    \end{tabular}
    \begin{tabular}[b]{c}
      \includegraphics[width=0.48\textwidth]{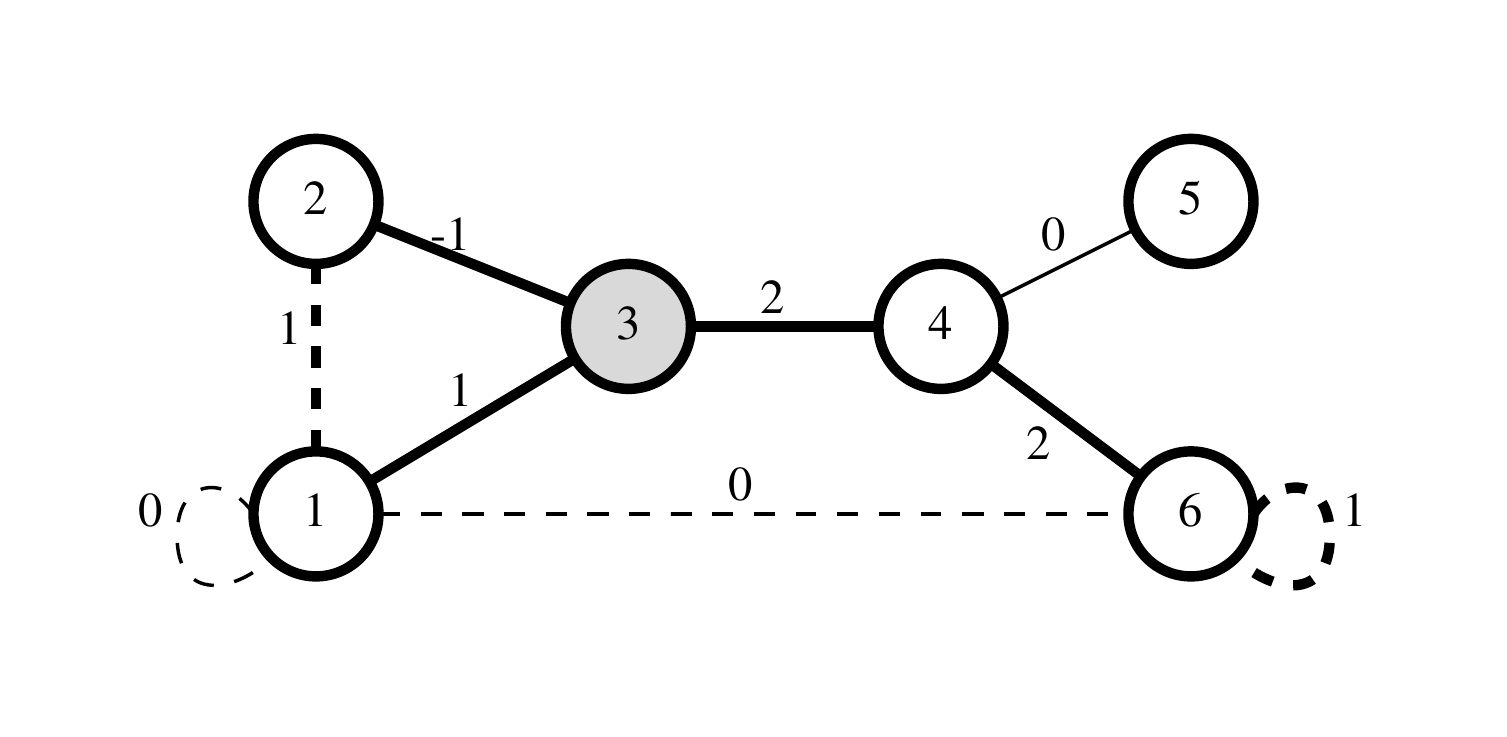}\\
      \small (d)
    \end{tabular}
    \caption{\label{fig:basis} The basis constructed from the cycles
      in Fig.~\ref{fig:cycles} and the corresponding values of the
      respective vector ${\bf y}_i\in{\mathbb{R}}^{|E|}$. There is one
      basis vector corresponding to the clean edge $\{1,6\}$ in (a),
      and three basis vectors corresponding to pairs of dirty edges:
      $\{1\}$ and $\{1,2\}$ in (b), $\{1\}$ and $\{6\}$ in (c), and
      $\{1,2\}$ and $\{6\}$ in (d). Each of these basis vectors
      multiplied by matrix ${\bf A}$ of Tab.~\ref{tab:example1} yields
      zero and therefore the basis vectors are in the null space ${\rm
        Null}({\bf A})$. The basis vectors above are also given by the
      bottom rows of Tab.~\ref{tab:example1}. All of the introduced
      graph cycles defined by edges with non-zero weights are of even
      length: (a) $1-3-4-6(-1-\ldots)$, (b) $1-1-2-3(-1-\ldots)$, (c)
      $1-1-3-4-6-6-4-3(-1-\ldots)$, and (d)
      $1-2-3-4-6-6-4-3(-1-\ldots)$; notice that edges with the weight
      of $\pm 2$ are traversed twice in a graph cycle, once in each
      direction.}
  \end{center}
\end{figure}

\section{Experimental Evaluation}
To demonstrate the scalability of the method described in this paper
we perform two sampling experiments on seven publicly available sparse
real-world networks. These networks are all examples of
\emph{recommendation datasets} (note that this does not limit the
generality of the discussion). In recommendation data, a \emph{user}
provides a \emph{rating} for a given \emph{item}, i.e., the data items
are triplets of the form (user, item, rating). The properties of the
networks are presented in Table~\ref{tab:datasets}. The table shows
the dimensions of the networks in terms of the number of rows (users)
and columns (items) in the data matrix. The density of all networks is
very low, meaning that the networks are sparse. The table also shows
the number of edges and vertices in the network and the dimensionality
of the null space.

We performed the following preprocessing steps for the networks. (i)
Duplicated edges were removed. (ii) For \texttt{Last.fm} and
\texttt{Tasteprofile} all ratings above 2500 and 20, respectively,
were discarded. (iii) For \texttt{BookCrossing} we only used explicit
ratings (i.e., nonzero ratings), rows with book id:s containing the
symbol `?'~were discarded and the symbols `\textbackslash', `=' as
well as blank spaces were removed from the book id:s. For each
network, the edge weights $w(e)$ were scaled to the interval $(0, 1)$.

When investigating varying properties of recommendation datasets it
makes sense to place certain restrictions on (i) the sum of the
ratings given by a user to all items and (ii) the sum of all ratings
received by an item from all users.

In the \textbf{first experiment} we sample networks where the vertex
weights are preserved exactly, while the edge weights $w(e)$ are
allowed to vary on the interval $(0, 1)$, i.e., the range of the edge
weights in the original network. In the context of the recommender
systems this means that the total ratings given by a user and the
total ratings received by an item are both preserved exactly (i.e.,
the ratings for a given user are just allocated differently).

In the \textbf{second experiment} we sample networks where both edge
and vertex weights are allowed to vary. The edge weights are again
constrained to the interval $(0, 1)$ while the vertex weights $W(v)$
are constrained to the interval $\lbrack 0.9 \cdot W(v), 1.1 \cdot
W(v)\rbrack$ for each vertex $v \in V$ in the original network. In the
context of recommender systems this means that the total ratings given
by a user and the total ratings received by an item cannot vary more
than $\pm 10\%$ from the value observed in the original dataset.

In both experiments we consider the scalability of the sampling method
presented in this paper in terms of the convergence rate of the
sampler. Studying the convergence of a Markov chain is a nontrivial
problem and we here consider the convergence in terms of how the
Frobenius norm
\begin{equation}
  {\|\mathbf{w} - \mathbf{w}_j^* \|}_F = {\left( \sum_{i=1}^{n} {\left( \mathbf{w}(i) - \mathbf{w}^*_j(i) \right)}^2 \right)}^{1/2}
\end{equation}
between the edge weight vector of the observed network ($\mathbf{w}$)
and the $j$th sample from the sampler ($\mathbf{w}_j^*$) evolves. The
sampler presented in this paper is implemented in C, as an extension
to R \cite{Rproject} and is freely available for download from
\url{http://github.com/edahelsinki/cyclesampler}.

\subsection{Results}
In the experiments we set a target of 100 000 samples with a sampling
time cut-off of 48 hours. The full number of samples was in both
experiments obtained for all networks except for \texttt{MovieLens
  20M} and \texttt{TasteProfile}. For these networks we obtained 13
900 and 4 500 samples in the first experiment and 12 550 and 4 050
samples in the second experiment, respectively.

The initialisation and sampling times (both in seconds) of the sampler
are presented in Table~\ref{tab:datasets}, recorded on a standard
laptop equipped with a quad-core \unit[2.6]{GHz} Intel Core i7
processor and \unit[20]{Gb} of RAM, running a 64-bit version of R
(v. 3.5.1) on Linux. The initialisation time ($t_\mathrm{init}$) is
the time required to set up the sampler, which consists of determining
the spanning tree and identifying the cycles. The sampling time
($t_\mathrm{sample}$) for a particular network is the time required to
take a number of steps equal to the dimensionality of its null space
$|C|$, i.e., $|E| - |V| + 1$ in experiment 1 and $|E| + 1$ in
experiment 2 (all of our networks are bipartite). The subscripts
\emph{1} and \emph{2} are used to denote the initialisation times for
experiment 1 and 2, respectively. The initialisation and sampling
times increase as the dimensionality of the null space of the network
increases. This is also reflected in the initialisation and sampling
times for experiment 2, where the dimensionality of the null space is
higher due to the addition of one self-loop per vertex required to
preserve vertex weights on an interval (see
Section~\ref{sec:vertexconstraints}).

The initialisation time is about a second for small networks
(\texttt{Last.fm}, \texttt{MovieLens 100k}) and less than 10 minutes
even for the \texttt{TasteProfile} network with tens of millions of
edges. Similarly, the time needed to produce a sample ranges from a
fraction of a second for the small networks to about 1.5 minutes for
the largest network.

The convergence results from experiments 1 and 2 are shown in
Figure~\ref{fig:convergence}, showing the evolution of the Frobenius
norm between the starting state and the sampled state.  The Frobenius
norm is normalised, for each network, to the interval $\lbrack0,
1\rbrack$ so that $0$ corresponds to the starting state and $1$ to the
maximum value of the norm. (The curves are jittered in the vertical
direction for visualisation purposes to prevent overplotting.) Also,
for visualisation purposes the curves are plotted on a logarithmic
scale, and the plotted points are taken at logarithmically spaced
intervals. One step on the $x$-axis corresponds to a number of steps
equal to the the dimensionality of the null space of a given network,
which facilitates comparisons between the different networks.

The results from experiment 1, where the vertex weights are preserved
exactly, are shown in Figure~\ref{fig:convergence}(a). We notice that
in the order of 1 000$|C|$ steps are needed for the sampler to
converge for all datasets except for \texttt{Last.fm}, which has not
converged after 100 000$|C|$ steps.

The results from experiment 2, where the vertex weights are preserved
on an interval, are shown in Figure~\ref{fig:convergence}(b). The
sampler clearly converges more slowly for all datasets than in
experiment 1; the number of steps required for convergence appears to
be in the order of 10 000$|C|$ steps, which is approximately a tenfold
increase in number of steps compared to experiment 1. Here the sampler
has not yet converged for \texttt{Last.fm}, \texttt{MovieLens 20M} and
\texttt{TasteProfile}.

\begingroup
\squeezetable
\begin{table}[ht]
  \centering
  \caption{Properties of the networks. The networks are sorted in
    order of an increasing number of edges. The columns are as
    follows: \emph{rows} and \emph{columns} give the full size of the
    data matrix and \emph{density} is the number of nonzero
    entries. The number of edges, vertices and the dimensionality of
    the null space (in experiment 1) are given by $|E|$, $|V|$ and
    \emph{$|C| = |E|-|V|+1$}, respectively. In experiment 2 the
    dimensionality of the null space is $|E|+1$ due to the addition of
    one self-loop per vertex. The initialisation time for the sampler
    (e.g., finding the spanning tree and enumerating cycles) and the
    time needed to take a number of steps equal to the dimensionality
    of the null space of the network are shown in the columns
    $t_\mathrm{init}$ and $t_\mathrm{sample}$. The times are in
    seconds and the subscript \emph{1} refers to experiment 1 whereas
    the subscript \emph{2} refers to experiment 2.}
  \label{tab:datasets}
  \begin{ruledtabular}
    \begin{tabular}{r cccccccccc}
      & \textbf{rows} & \textbf{columns} & \textbf{density} & $\boldsymbol{|E|}$  & $\boldsymbol{|V|}$ & $\boldsymbol{|C|}$ & $\mathrm{\textbf{t}}_\mathrm{\textbf{init,1}}$ & $\mathrm{\textbf{t}}_\mathrm{\textbf{sample,1}}$ & $\mathrm{\textbf{t}}_\mathrm{\textbf{init,2}}$ & $\mathrm{\textbf{t}}_\mathrm{\textbf{sample,2}}$ \\
      \hline
      \textbf{Last.fm ${}^A$}        & $1.89 \times 10^{3}$ & $1.74 \times 10^{4}$ & $2.69 \times 10^{-3}$ & $8.86 \times 10^{4}$ & $1.93 \times 10^{4}$ & $6.93 \times 10^{4}$ & 0.85   & 0.03  & 1.09   & 0.05 \\ 
      \textbf{MovieLens 100k ${}^B$} & $9.43 \times 10^{2}$ & $1.68 \times 10^{3}$ & $6.30 \times 10^{-2}$ & $1.00 \times 10^{5}$ & $2.62 \times 10^{3}$ & $9.74 \times 10^{4}$ & 0.72   & 0.04  & 0.79   & 0.05 \\ 
      \textbf{BookCrossing ${}^C$}   & $7.78 \times 10^{4}$ & $1.86 \times 10^{5}$ & $3.00 \times 10^{-5}$ & $4.34 \times 10^{5}$ & $2.64 \times 10^{5}$ & $1.85 \times 10^{5}$ & 36.88  & 0.20  & 58.82  & 0.66 \\ 
      \textbf{FineFoods ${}^D$}      & $2.56 \times 10^{5}$ & $7.43 \times 10^{4}$ & $2.95 \times 10^{-5}$ & $5.61 \times 10^{5}$ & $3.30 \times 10^{5}$ & $2.53 \times 10^{5}$ & 62.83  & 0.27  & 91.59  & 0.84 \\ 
      \textbf{MovieLens 1M ${}^B$}   & $6.04 \times 10^{3}$ & $3.71 \times 10^{3}$ & $4.47 \times 10^{-2}$ & $1.00 \times 10^{6}$ & $9.75 \times 10^{3}$ & $9.90 \times 10^{5}$ & 7.48   & 0.61  & 8.12   & 0.67 \\ 
      \textbf{MovieLens 20M ${}^B$}  & $1.38 \times 10^{5}$ & $2.67 \times 10^{4}$ & $5.40 \times 10^{-3}$ & $2.00 \times 10^{7}$ & $1.65 \times 10^{5}$ & $1.98 \times 10^{7}$ & 164.05 & 23.41 & 165.05 & 15.40 \\ 
      \textbf{TasteProfile ${}^E$}   & $1.02 \times 10^{6}$ & $3.84 \times 10^{5}$ & $1.22 \times 10^{-4}$ & $4.77 \times 10^{7}$ & $1.40 \times 10^{6}$ & $4.63 \times 10^{7}$ & 479.00 & 85.45 & 573.51 & 80.20 \\ 
    \end{tabular}
  \end{ruledtabular}
  \begin{tabular}{ll}
    ${}^A$ & \url{http://files.grouplens.org/datasets/hetrec2011/hetrec2011-lastfm-readme.txt},  \url{http://www.lastfm.com} and \cite{Cantador:RecSys2011} \\
    ${}^B$ & \url{http://grouplens.org/datasets/movielens/} and \cite{harper2016movielens} \\
    ${}^C$ & \url{http://www2.informatik.uni-freiburg.de/~cziegler/BX/} and \cite{ziegler2005improving}\\
    ${}^D$ & \url{https://snap.stanford.edu/data/web-FineFoods.html} and \cite{mcauley2013amateurs}  \\
    ${}^E$ &  \url{http://labrosa.ee.columbia.edu/millionsong/tasteprofile} and \cite{Bertin-Mahieux2011} \\
  \end{tabular}
\end{table}
\endgroup

\begin{figure}
  \begin{center}
    \begin{tabular}[b]{c}
      \includegraphics[width=0.48\textwidth]{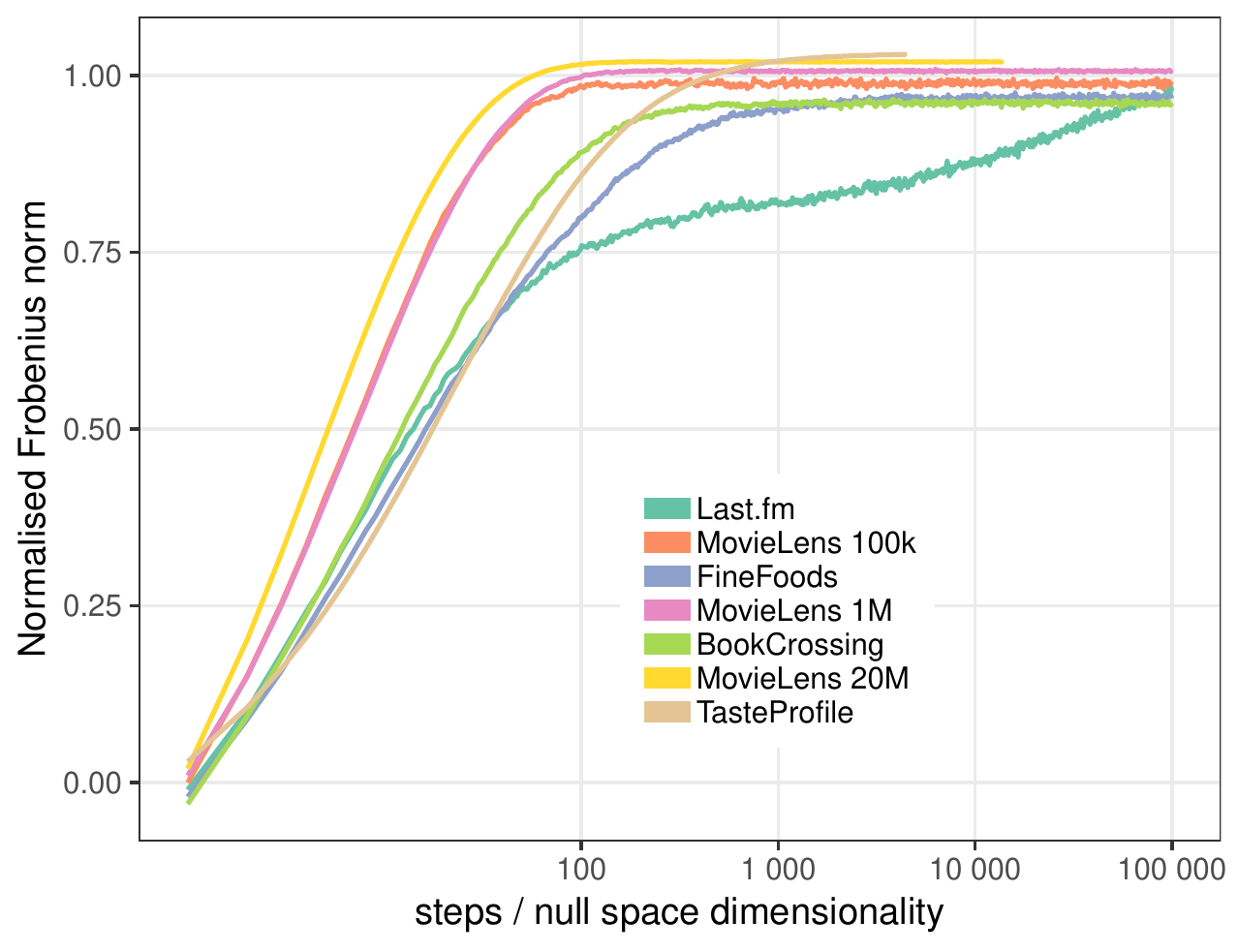}\\
      \small (a)
    \end{tabular}
    \begin{tabular}[b]{c}
      \includegraphics[width=0.48\textwidth]{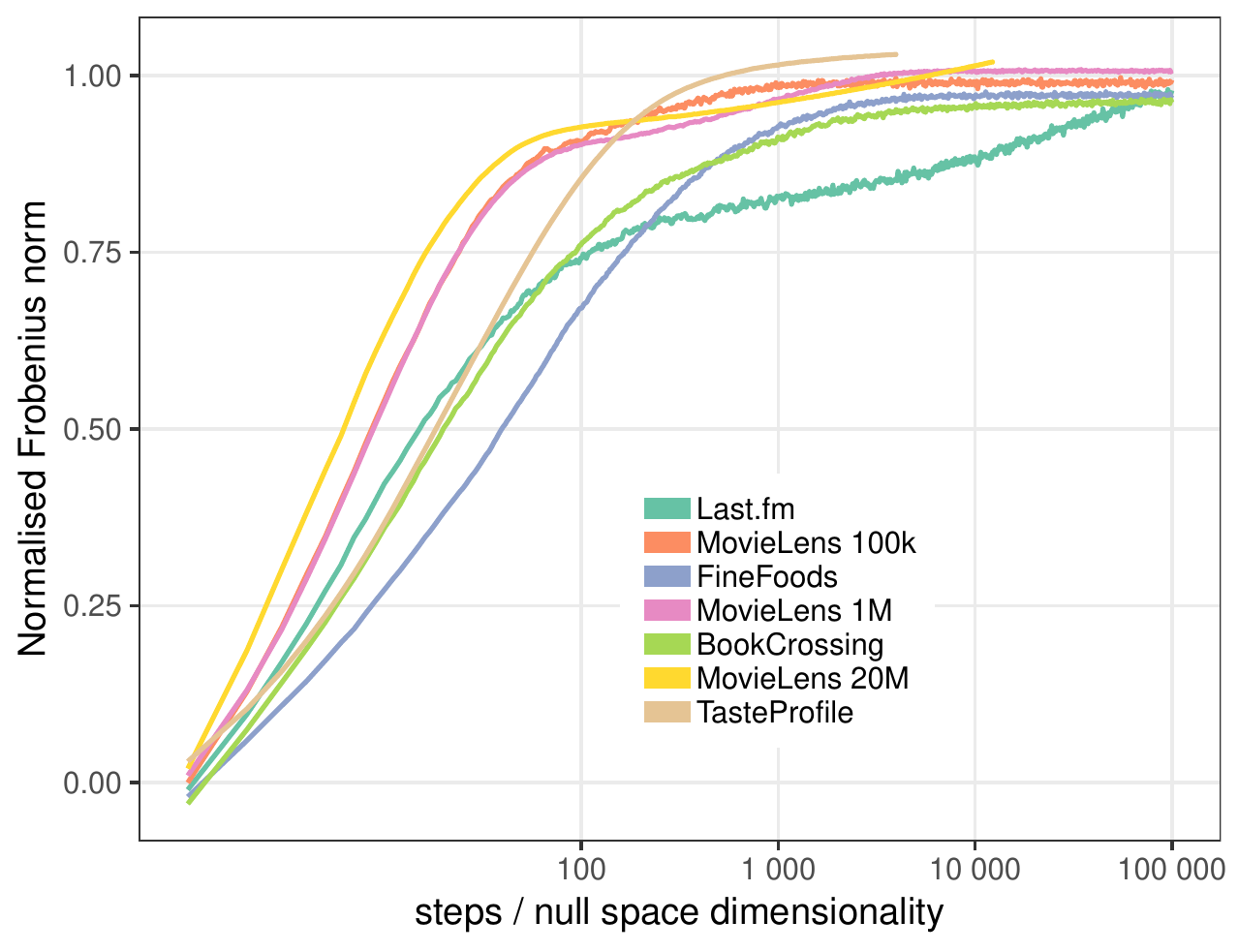}\\
      \small (b)
    \end{tabular}
    \caption{Evolution of the Frobenius norm between the starting
      state and the current state of the sampler. Figure (a) shows the
      results for experiment 1, where vertex weights are preserved
      exactly. Figure (b) shows the results for experiment 2, where
      vertex weights are allowed to vary on an interval.}
    \label{fig:convergence}
  \end{center}
\end{figure}

\section{Conclusions}
Sampling networks has important applications in multiple domains where
it is required to investigate and understand the significance of
different phenomena described by the structure of the network. Many
such networks are often sparse and large and consequently generating
surrogate networks adhering to specific constraints is a difficult
problem. In this paper we introduced \texttt{CycleSampler}; a novel
Markov chain Monte Carlo method that allows sampling of both
undirected and directed networks with interval constraints on both
edge and node weights.

The presented method provides an efficient means for sampling large
networks and we provided an empirical evaluation demonstrating that
the method scales to large sparse real-life networks. We believe that
the \texttt{CycleSampler}-method has applications in many domains and
we also release an open-source implementation of the method as an
R-package.

\acknowledgments{This work was funded by the Academy of Finland
  (decisions 326280 and 326339) and Tekes (Revolution of Knowledge
  Work). We acknowledge the computational resources provided by the
  Finnish Grid and Cloud Infrastructure \cite{fcgi} (persistent
  identifier urn:nbn:fi:research-infras-2016072533).}

\bibliographystyle{apsrev4-1}
\bibliography{paper}

\end{document}